\documentclass[12pt,draftcls,onecolumn]{IEEEtran}

\makeatletter
\def\ps@headings{%
\def\@oddhead{\mbox{}\scriptsize\rightmark \hfil \thepage}%
\def\@evenhead{\scriptsize\thepage \hfil \leftmark\mbox{}}%
\def\@oddfoot{}%
\def\@evenfoot{}}
\makeatother
\usepackage{amsfonts}
\usepackage{mathrsfs}
\usepackage{amsfonts}
\usepackage{amssymb}
\usepackage{graphicx}
\usepackage{epsfig}
\usepackage{epstopdf}
\usepackage{psfrag}
\usepackage{amsmath}
\usepackage{array}
\usepackage{subfigure}
\usepackage{cite,graphicx,amsmath,amssymb,color}
\usepackage{algorithmic}
\usepackage{algorithm}
\usepackage{algorithm}
\usepackage{algorithmic}
\usepackage{stmaryrd}
\usepackage{multirow}
\usepackage{subfig}
\usepackage{graphicx,times,amsmath} 
\usepackage{url}

\IEEEoverridecommandlockouts
\begin{document}
\bibliographystyle{IEEEtran}

\title{Uncoded Placement Optimization for Coded Delivery}

\author{Sian Jin,~\IEEEmembership{Student~Member,~IEEE}\thanks{S. Jin, Y. Cui and H. Liu are with Shanghai Jiao Tong University, China. G. Caire is with Technical University of Berlin, Germany. This
paper was presented in part at IEEE WiOpt 2018~\cite{conference}.}, \ Ying Cui,~\IEEEmembership{Member,~IEEE}, \\ Hui Liu,~\IEEEmembership{Fellow,~IEEE},  \ Giuseppe Caire,~\IEEEmembership{Fellow,~IEEE}}

\maketitle
\newtheorem{Thm}{Theorem}
\newtheorem{Lem}{Lemma}
\newtheorem{Sta}{Statement}
\newtheorem{Cor}{Corollary}
\newtheorem{Def}{Definition}
\newtheorem{Exam}{Example}
\newtheorem{Alg}{Algorithm}
\newtheorem{Sch}{Scheme}
\newtheorem{Prob}{Problem}
\newtheorem{Rem}{Remark}
\newtheorem{Proof}{Proof}
\newtheorem{Asump}{Assumption}
\newtheorem{Subp}{Subproblem}
\newtheorem{Con}{Condition}

\vspace{-1cm}

\begin{abstract}
We consider the classical coded caching problem as defined by Maddah-Ali and Niesen, where a server with a library of $N$  files  of equal size is connected to $K$ users via a shared error-free    link.
Each user is  equipped with a cache with capacity of    $M$ files.
The goal  is to design    a static content placement and delivery  scheme  such that the average load over the shared link is minimized.
Existing coded caching schemes fail to simultaneously achieve efficient content placement for non-uniform file popularity and efficient content delivery in the presence of common requests, and hence may not achieve desirable average load under a non-uniform, possibly very skewed, popularity distribution.
In addition, existing coded caching schemes usually require the splitting of a file into a large number of subfiles, i.e., high subpacketization level.
To address the above two challenges, we first present a class of centralized coded caching schemes consisting of a general content placement strategy specified by a file partition parameter, enabling efficient and flexible content placement, and a specific content delivery strategy, enabling load reduction by exploiting common requests of different users.
For the proposed class of schemes, we consider two cases for the optimization of the file partition parameter, depending on whether a large subpacketization level is allowed or not.
In the case of an unrestricted subpacketization level, we formulate the coded caching optimization in order to  minimize the average load under an arbitrary file popularity.
A direct formulation of the problem involves $N2^K$ variables.
By imposing some additional conditions, the problem is reduced to a linear program with $N(K+1)$ variables under an arbitrary file popularity and with $K+1$ variables under the  uniform file popularity.
We can recover Yu {\em et al.}'s optimal scheme for the uniform file popularity as an optimal solution of our  problem.
When a low subpacketization level is desired, we introduce  a subpacketization level constraint involving the $\ell_0$ norm  for each file.
Again, by imposing  the same additional conditions, we can  simplify   the problem to a difference of two convex functions  (DC) problem with $N(K+1)$ variables  that can be efficiently solved.
\end{abstract}


\section{Introduction}
The rapid proliferation of smart mobile devices has triggered
an unprecedented growth of the global mobile data traffic, with a predicted nearly seven-fold
increase between 2016 and 2021 \cite{Cisco}. In order to support such
dramatic growth of wireless data traffic, caching and multicasting  have been recently
proposed as two promising approaches for massive content delivery in wireless networks.
Joint design of the two promising techniques is expected to achieve superior performance for massive content delivery in wireless networks.
In  \cite{multicast2014,CuiTWC1,CuiTWC2},  the authors  consider joint design of traditional uncoded caching and multicasting, the gain of which mainly derives from making content available locally and serving multiple
requests of the same contents concurrently.

Recently, a new class of caching schemes for content placement in user caches, referred to as {\em coded caching}~\cite{AliFundamental},  have received significant interest.
In~\cite{AliFundamental}, Maddah-Ali and Niesen consider a system with one server connected  through a shared error-free link to $K$ users.
The server has a database of $N$ files (of  $F$ data units), and each user  has an isolated cache
memory containing up to $M$ files.
They formulate a caching problem consisting of two phases, namely,
a content placement phase and a content delivery phase.
The content placement is performed once, before operating the network, and independently of the user requests.
Then, the users place requests in rounds, and at each round the server responds with
a   multicast  message  constructed by coded multicast XOR operations  that satisfies all user requests simultaneously.
The goal of~\cite{AliFundamental} is to reduce the worst-case (over all possible requests) load of the shared link in the delivery phase.
In~\cite{jin},    we  consider a   different class of  centralized  coded caching schemes   specified by a  general file partition parameter, and optimize the parameter to minimize the average  (over random requests)  load  within the class under an  arbitrary file popularity.
In~\cite{Wei_Yu}, the parameter-based coded caching design approach in~\cite{jin} is generalized to minimize the average load in a heterogeneous setting with nonuniform   cache size and file size  under an arbitrary file popularity.
In~\cite{YuQian},  Yu {\em et al.} propose a centralized coded caching scheme where the delivery strategy   exploits the chance of load reduction  in common requests of different users  and prove its  information theoretic  optimality   for the worst-case  load  and average load under the uniform file  popularity.

Note that the delivery strategies in~\cite{AliFundamental,jin,Wei_Yu} do not capture the opportunity of load reduction in  common requests of different users, and the placement strategies in~\cite{AliFundamental} and~\cite{YuQian} allocate the same fraction of memory to each file without reflecting popularity difference  of files.
Therefore, the coded caching schemes in~\cite{AliFundamental,jin,Wei_Yu, YuQian} may not achieve desirable average load under a non-uniform, possibly very skewed, popularity distribution.
At the moment, a general  optimality result for random requests with an arbitrary file popularity is not known.

Another limitation of~\cite{AliFundamental,jin,Wei_Yu, YuQian} is the  issue of high subpacketization  level,  i.e.,  the  number of  non-overlapping  subfiles  for  each  file is large.
In~\cite{Tang,Shanmugam17,cheng},  the authors  tackle   the subpacketization  level   issue for  centralized coded caching.
Specifically,   in~\cite{Tang},    Tang {\em et al.} connect coded caching to   resolvable combinatorial designs and propose a  centralized coded caching  scheme  where  the  subpacketization  level  is exponential with respect to  (w.r.t.)  the number of users  but  with a smaller exponent constant than in  the  centralized coded caching  scheme of~\cite{AliFundamental} at the cost of a marginal increase in the worst-case load.
In~\cite{Shanmugam17},  Shanmugam {\em et al.} connect coded caching to  Ruzsa-Szem{\'{e}}redi graphs and   show  the existence of  a   centralized  coded caching scheme  where the subpacketization  level  grows linearly with the number of users.
However, such scheme exists only when the number of users is impractically large.
In~\cite{cheng}, the authors propose a  centralized coded caching scheme with low subpacketization  level    based on Pareto-optimal placement delivery array   (PDA).
Note that  the   centralized  coded caching schemes in~\cite{Tang,Shanmugam17,cheng}  addressing the subpacketization  level   issue are    applicable only for certain system parameters (e.g., the number of users, the number of files,   cache size, etc.).
Furthermore,   the  centralized   coded caching schemes in~\cite{Tang,Shanmugam17,cheng}    are  based on combinatorial designs
and do not explicitly solve any optimization problem under subpacketization level constraints.

In this paper, we would like to address the above challenges  in the same centralized setting as in~\cite{AliFundamental, jin,Wei_Yu, YuQian,Tang,Shanmugam17,cheng},  with the focus on minimizing  the average load under an arbitrary file popularity   in two  cases, namely, the case without considering the subpacketization  level  issue and the  case  considering the subpacketization  level  issue.
We  present a class of coded caching schemes  consisting of a general content placement  strategy specified by  a  file partition parameter, enabling efficient and flexible content placement, and  a specific content delivery  strategy, enabling  load reduction by exploiting  common requests  of different users.
Then,  we focus on the average load minimization irrespectively of the subpacketization  level  issue. In this case,  we  formulate the   coded caching  optimization problem over the considered class of schemes   to minimize the average load under an arbitrary file popularity.
The  average load expression is not tractable due to the complex delivery strategy.
Therefore,  we impose some  additional  conditions on the  parameter  to  simplify  the  average load expression  under an arbitrary  file popularity and  the  uniform file popularity respectively, by connecting the file request event   to the ``balls into bins'' problem.
Based on  the simplified  expressions, we transform the original optimization problem with  $N2^K$ variables    into a   linear program with $N(K+1)$ variables under an arbitrary  file popularity  and a  linear program with $K+1$  variables under the  uniform file popularity, which are much easier to solve than the original problem.
We also show that Yu {\em et al.}'s centralized coded caching scheme corresponds to  an optimal solution  of our problem,  thus implying that the imposed conditions incur no loss of optimality for the uniform file popularity.
Next, we  focus on  the average load minimization considering the  subpacketization  level  issue. In this case, we first  formulate the   coded caching  optimization problem over the considered class of schemes  to minimize the average load under an arbitrary file popularity  subject to   subpacketization  level  constraints in terms of  the $\ell_0$-norm of the file partition parameter.
To the best of our knowledge, this is the first work  explicitly considering   subpacketization  level  constraints in  the optimization of   coded caching design.
By  imposing  the same additional  conditions as  before  and using the exact difference of two convex functions  (DC)   reformulation method in~\cite{dc2017},   we convert the original problem with $N2^K$ variables   into a simplified DC problem with $N(K+1)$ variables.
Then, we use a    DC algorithm to solve the simplified DC problem.
Numerical results reveal  that   the imposed conditions do not affect the optimality of the original problem    under an arbitrary file popularity  in both   cases.
Furthermore,  our  numerical  results  demonstrate that the optimized coded caching scheme  without considering the subpacketization  level  constraints  outperforms those in~\cite{AliFundamental, jin, YuQian}  in terms of the average load, and   the optimized coded caching scheme  considering the subpacketization  level  constraints outperforms those in~\cite{Tang} and~\cite{cheng} in terms of  both  the average load and application region.

\section{Centralized  Coded Caching}\label{Sec:scheme}
\subsection{Problem Setting}\label{Sec:setting}
As in~\cite{AliFundamental, jin,Wei_Yu, YuQian,Tang,Shanmugam17,cheng}, we consider a system with one server connected through a shared error-free link to $K \in \mathbb N_{>0}$  users  (see Fig.~\ref{fig:system_model}),
where $\mathbb N_{>0}$   denotes the set of all positive integers.
The server has access to a  library  of $N \in \mathbb N_{>0}$
files, denoted by $W_1, \ldots ,W_N$, each consisting of $F \in \mathbb N_{>0}$ indivisible data units.
Let  $\mathcal{N}\triangleq \{1,2,\ldots ,N\}$ and $\mathcal{K} \triangleq \{1,2, \ldots K\}$ denote the set of file indices and the set of user indices, respectively.
Each user  has an isolated cache memory of $MF$ data units, for some real number $M \in [0,N]$.
Let $Z_k$ denote the  cache content for user $k$.
The system operates in two phases, i.e., a placement phase and a delivery phase~\cite{AliFundamental}.
In the placement phase, each user is  able to fill the content of its cache using the library  of $N$ files.
In the delivery phase, each user randomly and independently requests one file in $\mathcal N$ according to file popularity distribution  $\mathbf p \triangleq \left(p_n\right)_{n=1}^{N}$, where $p_n$ denotes the probability of a user requesting file $W_n$ and $\sum_{n=1}^{N}p_n =1$. Without loss of generality, we assume $p_1 \geq p_2 \geq \ldots \geq p_N$.
Let $D_{k}\in\mathcal N$ denote the index of the file requested by user $k \in \mathcal{K} $, and
let $\mathbf D\triangleq \left(D_1, \cdots, D_K\right)\in \mathcal N^K$ denote the requests of  all the $K$ users.
The server replies to these $K$ requests by sending   messages  over the shared link, which are observed by all the $K$ users.
Each user  should be   able to recover its requested file from the messages received over the shared link and its cache content.
Our goal is to  minimize  the average load  of the shared link  under an arbitrary file popularity.

\begin{figure}
\begin{center}
  {\resizebox{9cm}{!}{\includegraphics{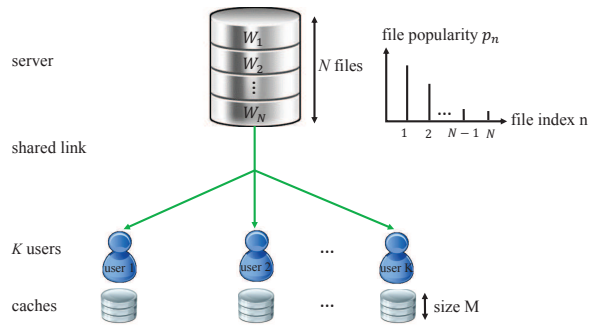}}}
         \caption{\small{Problem setup for coded caching~\cite{jin}.}
         }\label{fig:system_model}
\end{center}
\end{figure}
\subsection{Centralized  Coded Caching Scheme}
In this  part,  we  present   a  class of  centralized  coded caching schemes    utilizing  a  general  uncoded placement strategy and  a specific coded delivery strategy, which  are specified  by a  general  file partition parameter, as  summarized in  Alg.~\ref{alg:col}.
This uncoded placement strategy was introduced in our previous work~\cite{jin} and it is repeated here for completeness.
For all $n \in \mathcal N$, file $W_n$  is  partitioned into  $2^K$ nonoverlapping   subfiles $W_{n,\mathcal{S}}$, $\mathcal{S} \subseteq \mathcal{K}$, i.e., $W_n=\left\{W_{n,\mathcal{S}}: \mathcal{S} \subseteq \mathcal{K}\right\}.$
If the number of data units in a subfile is zero, then there is no need to consider this subfile.
Thus,  $2^K$ is the maximum  number of non-overlapping subfiles  of a file.
We say subfile  $W_{n,\mathcal{S}}$ is of type $s$ if $|\mathcal{S}|=s$~\cite{jin}.
User $k$ stores $W_{n,\mathcal{S}}$, $n \in \mathcal N ,k \in \mathcal{S}, \mathcal{S} \subseteq \mathcal{K}$ in its cache, i.e., $Z_k=\left\{W_{n,\mathcal{S}}:n \in \mathcal N ,k \in \mathcal{S}, \mathcal{S} \subseteq \mathcal{K}\right\}.$
Let $x_{n,\mathcal{S}}$ denote  the   size of   subfile   $W_{n,\mathcal{S}}$, normalized by the file size $F$.
Denote  $\mathbf{x}_n \triangleq \left(x_{n,\mathcal{S}}\right)_{\mathcal{S} \subseteq \mathcal{K}}$.
Let   $\mathbf{x} \triangleq (\mathbf{x}_n)_{n \in \mathcal{N}}$   denote the  file partition parameter, which   will be optimized  to minimize the average  load in Section~\ref{Sec:optimization} and Section~IV.
Thus, $\mathbf{x}$ satisfies
\begin{align}
&0 \leq x_{n, \mathcal{S}} \leq 1 , \quad \forall \mathcal{S} \subseteq \mathcal{K}, \ n \in \mathcal{N},\label{eqn:X_range}\\
&\sum_{s=0}^{K} \sum_{{\cal S} \subseteq {\cal K} : |{\cal S}| = s}  x_{n, \mathcal{S}}=1 , \quad  \forall n \in \mathcal{N}, \label{eqn:X_sum}\\
&\sum_{n=1}^{N}\sum_{s=1}^{K}\sum_{{\cal S} \subseteq {\cal K} : |{\cal S}| = s, k \in {\cal S}}  x_{n, \mathcal{S}} \leq M, \quad  \forall k \in \mathcal K, \label{eqn:memory_constraint}
\end{align}
where \eqref{eqn:X_range} and  \eqref{eqn:X_sum} represent the file partition constraints and \eqref{eqn:memory_constraint} represents the cache memory constraint.

The coded delivery strategy is an extension of that  in~\cite{YuQian}.
For all $\mathbf D \in \mathcal N^K$, let $\underline{\mathcal D}(\mathbf D)$ denote the set of distinct files in $\mathbf D$.
For all $n \in \underline{\mathcal D}(\mathbf D)$, the server arbitrarily selects user $k_n \in \mathcal K$ such that $D_{k_n}=n$.
Let  $\underline{\mathcal{K}}(\mathbf D) \triangleq \{k_n:n \in \underline{\mathcal D}(\mathbf D)\}$ denote the set of representative users that request $|\underline{\mathcal D}(\mathbf D)|$ different files.
Each user $k \in \mathcal{S}$ requests  subfile $W_{D_k,\mathcal{S}\setminus \{k\}}$, for all subset $\mathcal{S} \subseteq \mathcal K$.
The server  broadcasts coded-multicast message $\oplus_{k \in \mathcal{S}} W_{D_k,\mathcal{S}\setminus \{k\}}$\footnote{Note that  in~\cite{YuQian}, since all files are partitioned into subfiles of type $t \in \{0,1,\ldots, K\}$, only coded-multicast messages $\oplus_{k \in \mathcal{S}} W_{D_k,\mathcal{S}\setminus \{k\}}$ satisfying $\mathcal{S} \subseteq \mathcal K$, $\mathcal{S} \cap \underline{\mathcal{K}}(\mathbf D)\neq \emptyset$ and $|\mathcal{S}|=t+1$ are transimitted.} for all subset $\mathcal{S} \subseteq \mathcal K$ that satisfies $\mathcal{S} \cap \underline{\mathcal{K}}(\mathbf D)\neq \emptyset$,
and all subfiles in the coded-multicast message are   zero-padded to the length of the longest subfile.
By Lemma~1 of~\cite{YuQian}, we can conclude that each user  can decode the requested file based on  the received coded-multicast messages and  the contents stored in its cache.

\begin{Rem}[Comparison with Existing Coded Caching Schemes]
The uncoded  placement  strategy  in this paper  is more general   than that in~\cite{AliFundamental,YuQian, Tang,Shanmugam17,cheng}, and it can be optimized to minimize  the average load  under an arbitrary file popularity (see Section~\ref{Sec:optimization}  and Section~\ref{Sec:L0})~\cite{jin}.
The coded delivery  strategy  in this paper  is more efficient  than that  in~\cite{AliFundamental,jin,Wei_Yu,Tang,Shanmugam17,cheng},
since it avoids transmitting the redundant coded-multicast messages
$\oplus_{k \in \mathcal{S}} W_{D_k,\mathcal{S}\setminus \{k\}}$, $\mathcal S \subseteq \mathcal K \setminus \underline{\mathcal{K}}(\mathbf D)$ in the presence of common requests~\cite{YuQian}.
\end{Rem}

\begin{algorithm}[h]
\caption{Parameter-based Centralized Coded Caching}
\small{\textbf{placement strategy}}
\begin{algorithmic}[1]
\FORALL {$k \in \mathcal K$}
  \STATE $Z_k \leftarrow \left\{W_{n,\mathcal{S}}:n \in \mathcal N ,k \in \mathcal{S}, \mathcal{S} \subseteq \mathcal{K} \right\}$
\ENDFOR
\end{algorithmic}
\small{\textbf{delivery strategy}}
\begin{algorithmic}[1]
\FOR {$s=K,K-1,\cdots, 1$}
  \FOR {$\mathcal{S} \subseteq \mathcal{K}:|\mathcal{S}|=s,  \mathcal{S} \cap \underline{\mathcal{K}}(\mathbf D)\neq \emptyset$}
     \STATE server sends $\oplus_{k \in \mathcal{S}} W_{D_k,\mathcal{S}\setminus \{k\}}$
  \ENDFOR
\ENDFOR
\end{algorithmic}\label{alg:col}
\end{algorithm}

\subsection{Average Load}
Let $R_{\rm avg}(K,N,M,\mathbf{x})$ denote the average load for serving the $K$ users with cache size $M$ under a given file partition parameter $\mathbf{x}$, where the average is taken over random requests $\mathbf D$  for $N$ files, according to an arbitrary file popularity distribution $\mathbf p$.
By Alg.~\ref{alg:col}, we have
\begin{align}
R_{\rm avg}(K,N,M,\mathbf{x}) =\sum_{\mathbf{d} \in \mathcal{N}^K} \left(\prod_{k=1}^{K}p_{d_k}\right)
\sum_{\mathcal{S}    \subseteq \mathcal{K}:  \mathcal{S} \cap \underline{\mathcal{K}}(\mathbf D)\neq \emptyset} \max_{k \in \mathcal{S}} x_{d_k,\mathcal{S}\setminus\{k\}}, \label{eqn:average_load_1}
\end{align}
where $\mathbf{d} \triangleq (d_1, \ldots, d_K) \in \mathcal N^K$  and $\underset {k \in \mathcal{S}}\max \ x_{d_k,\mathcal{S}\setminus\{k\}}$ is  the length of the coded message $\oplus_{k \in \mathcal{S}} W_{d_k,\mathcal{S}\setminus \{k\}}$, normalized by the file size $F$.

From \eqref{eqn:average_load_1}, we can  observe  that  the file partition parameter  $\mathbf{x}$ fundamentally affects the average load $R_{\rm avg}(K,N,M,\mathbf{x})$.
In Section~\ref{Sec:optimization}, we would like to find an optimal file partition parameter  to minimize  the  average load in \eqref{eqn:average_load_1}.
The optimal file partition parameter may correspond to  high subpacketization  level.
In fact, if for some $n$, all the elements $x_{n, S}$ are non-zero, it means that file $W_n$ is divided  into  $2^K$ subfiles, which is exponential with the number of users  $K$.
For systems with even a moderate number of users and files of practical length, such partition becomes quickly impossible.
For example, for $K = 50$ and the size of each indivisible data unit  equal to $1$ Byte, we need files with size larger than $1$ Petabyte.
To avoid high subpacketization  level,    in Section~\ref{Sec:L0},  we would like to find an optimal file partition parameter  to minimize  the  average load in \eqref{eqn:average_load_1} under  subpacketization  level  constraints.

\section{Average Load Minimization without Subpacketization Level Constraint} \label{Sec:optimization}
In this section, we  consider the minimization of the average load without any restriction on the subpacketization level.

\subsection{Problem Formulation}\label{Sub:formulation}
We would like to minimize the average load under  the file partition constraints in \eqref{eqn:X_range}  and   \eqref{eqn:X_sum} as well as  the cache memory constraint in \eqref{eqn:memory_constraint}.
\begin{Prob}[Optimization for Arbitrary File Popularity]\label{Prob:original}
\begin{align}
R^*_{\rm avg}(K,N,M) \triangleq \min_{\mathbf{x}} \quad &  R_{\rm avg}(K,N,M,\mathbf{x})\nonumber \\
s.t. \quad  &\eqref{eqn:X_range}, \eqref{eqn:X_sum}, \eqref{eqn:memory_constraint},\nonumber
\end{align}
where $R_{\rm avg}(K,N,M,\mathbf{x})$ is given by \eqref{eqn:average_load_1}.
\end{Prob}

The objective function of Problem~\ref{Prob:original} is convex, as it is a positive weighted sum of convex piecewise linear functions.
In addition, the constraints of Problem~\ref{Prob:original} are linear.
Hence, Problem~\ref{Prob:original} is a convex optimization problem.
The number of variables in Problem~\ref{Prob:original} is $N 2^K$.
Thus,  the complexity   of  Problem~\ref{Prob:original}  is huge, especially when $K$ and $N$ are large.
In Section~\ref{Sub:arbitraty} and  Section~\ref{Sub:uniform}, we shall focus on deriving  simplified formulations for Problem~\ref{Prob:original} to facilitate low-complexity optimal solutions under an arbitrary  file  popularity distribution and the uniform  file  popularity distribution, respectively.

\subsection{Optimization for Arbitrary  File  Popularity}\label{Sub:arbitraty}
First,  we  present two   structural conditions  on the file partition parameter $\mathbf x$.
These conditions impose a restriction  on  the feasible region and  enable  a  simplification of Problem~\ref{Prob:original}. We cannot prove that the resulting solution is optimal w.r.t. Problem~\ref{Prob:original} due to its complex objective function,  but we can numerically verify the optimality of the resulting solution.

\begin{Con}[Symmetry w.r.t. Type]\label{Con:symmetry}
For all $n \in \mathcal{N}$ and $s \in \{0,1, \cdots, K\}$,    the  values of
$x_{n, \mathcal{S}}$,  $\mathcal{S}\subseteq \{\mathcal{\widehat{S}} \subseteq \mathcal{K}: |\mathcal{\widehat{S}}|=s\}$  are the  same.
\end{Con}

Recall that all subfiles in one coded-multicast message are  zero-padded to the length of the longest subfile in the coded-multicast message,
 causing the  ``bit waste" effect~\cite{jin}.
Thus, imposing Condition~\ref{Con:symmetry} can  reduce   the  variance of the lengths of messages involved in the coded-multicast XOR operations,   hence  addressing  ``bit waste" problem.
By Condition~\ref{Con:symmetry},   we can set
\begin{align}
x_{n,\mathcal{S}}=y_{n,s}, \quad  \forall \mathcal{S} \subseteq \mathcal{K}, \  n \in \mathcal{N}, \label{eqn:symmetry}
\end{align}
where $s=|\mathcal{S}| \in \{0,1,\cdots, K\}$.
Here,   $y_{n,s}$ can be viewed as  the size of each subfile of type $s$ in each file $W_n$, normalized by the file size $F$.
Let $\mathbf{y}_n \triangleq (y_{n,s})_{s \in \{0,1,\cdots, K\}}$ and $\mathbf{y} \triangleq (\mathbf{y}_n)_{n \in \mathcal{N}}$.

\begin{Con}[Monotonicity w.r.t.  Popularity]\label{Con:popularity}
For all  $n \in \{1,2,\ldots, N-1\}$ and $s \in \{1,2, \cdots, K\}$,  when $p_n \geq p_{n+1}$,
\begin{align}
y_{n,s} \geq  y_{n+1,s}. 
\end{align}
\end{Con}

Condition~\ref{Con:popularity} indicates that, for all  $n \in \{1,2,\ldots, N-1\}$ and $s \in \{1,2, \cdots, K\}$, when $p_n \geq p_{n+1}$,  the size of subfiles $W_{n,\mathcal{S}}$, $\mathcal{S}\subseteq \{\mathcal{\widehat{S}} \subseteq \mathcal{K}: |\mathcal{\widehat{S}}|=s\}$ is no smaller than that of subfiles $W_{n+1,\mathcal{S}}$, $\mathcal{S}\subseteq \{\mathcal{\widehat{S}} \subseteq \mathcal{K}: |\mathcal{\widehat{S}}|=s\}$.
Intuitively, imposing  Condition~\ref{Con:popularity}  can reduce the average load, by dedicating more memory to a more popular file.
Unlike in our previous work~\cite{jin},  due to the complex   objective function in \eqref{eqn:average_load_1}, we cannot show that imposing Conditions 1 and 2 maintains the optimality of the solution w.r.t. the original Problem~\ref{Prob:original}.
Later, in Section V, we provide numerical evidence    suggesting that indeed Conditions 1 and 2 do not involve any loss of optimality.

Next, we simplify Problem~\ref{Prob:original} under   Conditions~\ref{Con:symmetry} and~\ref{Con:popularity}.
\textcolor{black}{First, we introduce some notations.}
Consider the number of representative users $|\underline{\mathcal{K}}(\mathbf D)|=u$.
Let $\widetilde{D}_{u,\langle1\rangle} \leq \widetilde{D}_{u,\langle 2\rangle} \leq \ldots \leq \widetilde{D}_{u,\langle K-u \rangle}$ denote  $\left(D_k\right)_{k \in \mathcal K \setminus \underline{\mathcal{K}}(\mathbf D)}$ arranged  in ascending order, so that $\widetilde{D}_{u,\langle i\rangle}$ is the $i$-th smallest.
Let $P'_{i,u,n} \triangleq \Pr  \left[\widetilde{D}_{u,\langle i\rangle}=n\right]$, for all $i=1, \ldots, K-u.$
Note that the file request event  $\widetilde{D}_{u,\langle i\rangle}=n$  can be treated as  the ``balls into bins'' problem, i.e., $K$ balls are placed in an  i.i.d.  manner  into $N$ bins,   where bin $n$ is selected with probability $p_n$.
Let  $K_n$ denote  the number of users  requesting file  $n$.
Note that $\sum_{n=1}^{N}K_n=K$ and $\sum_{n=1}^{N}\mathbf 1 [K_{n}>0]=u$.
Let $A_n \triangleq \sum_{n'=1}^{n-1}\mathbf 1 [K_{n'}>0]$, $B_{n,1}  \triangleq \sum_{n'=1}^{n-1}K_{n'}$ and $B_{n,2} \triangleq \sum_{n'=n+1}^{N}K_{n'}$.
Let $\mathcal L_{A_n,1} \triangleq \left\{\mathcal L \subseteq \{1,2,\ldots,n-1\}: |\mathcal L|=A_n\right\}$ and  $\mathcal L_{A_n,2} \triangleq \left\{\mathcal L   \subseteq \{n+1,n+2,\ldots,N\}: |\mathcal L|=u-A_n-1\right\}$.
Let $$\mathcal A_{B_{n,1}, \mathcal L} \triangleq \left\{(\alpha_{n'})_{n' \in \mathcal L} \in \mathbb N_{>0}^{|\mathcal L|}: \sum_{n' \in \mathcal L} \alpha_{n'} =B_{n,1}\right\},$$  for all  $\mathcal L \in \mathcal L_{A_n,1}$  and
$$\mathcal A_{B_{n,2}, \mathcal L} \triangleq \left\{(\alpha_{n'})_{n' \in \mathcal L} \in \mathbb N_{>0}^{|\mathcal L|}: \sum_{n' \in \mathcal L} \alpha_{n'} =B_{n,2}\right\},$$  for all  $\mathcal L \in \mathcal L_{A_n,2}$.
Let $\mathbb N$ denote the set of all natural numbers.
Define ${K \choose m_1,m_2, K-m_1-m_2} \triangleq \frac{K!}{m_1!m_2!(K-m_1-m_2)!}$, where $m_1 \in \mathbb N$, $m_2 \in \mathbb N$ and $m_1+m_2 \leq K$.
\textcolor{black}{By using  results for the  ``balls into bins'' problem and considering Conditions~\ref{Con:symmetry} and~\ref{Con:popularity}, we have the following result.}

\begin{Lem}[Simplification \textcolor{black}{of Problem~\ref{Prob:original} for}  Arbitrary File Popularity]\label{Lem:simplification pop}
\textcolor{black}{Under   Conditions~\ref{Con:symmetry} and~\ref{Con:popularity},
Problem~\ref{Prob:original} can be  converted into:}
\begin{Prob}[Simplified  Problem  for Arbitrary File Popularity] \label{Prob:simplify_2}
\begin{align}
\widetilde{R}^*_{\rm avg}(K,N,M,\mathbf{y}) \triangleq  \min_{\mathbf{y}} \quad &  \widetilde{R}_{\rm avg}(K,N,M,\mathbf{y})  \nonumber \\
s.t. \quad  &y_{n,s} \geq  y_{n+1,s}, \quad  \forall  n \in \{1,2,\ldots, N-1\},  s \in \{1,2, \cdots, K\} \label{eqn:additional_cons} \\
&0 \leq y_{n, s} \leq 1 , \quad  \forall  s \in \{0,1,\cdots, K\}, \ n \in \mathcal{N}, \label{eqn:X_range_2} \\
&\sum_{s=0}^{K} {K \choose s} y_{n, s}=1 , \quad  \forall n \in \mathcal{N}, \label{eqn:X_sum_2}\\
&\sum_{n=1}^{N}\sum_{s=1}^{K} {K-1 \choose s-1} y_{n, s} \leq M, \label{eqn:memory_constraint_2}
\end{align}
\end{Prob}
where
\begin{align}
\widetilde{R}_{\rm avg}(K,N,M,\mathbf{y})
\triangleq& \sum_{s=1}^{K}{K \choose s}\sum_{n=1}^{N}\left(\left(\sum_{n'=n}^{N}p_{n'}\right)^s-\left(\sum_{n'=n+1}^{N}p_{n'}\right)^s\right) y_{n,s-1}\nonumber \\
-&\sum_{u=1}^{\min \{K,N\}}\sum_{s=1}^{K-u}{K-u \choose s}\sum_{i=1}^{K-u}{K-u-i \choose s-1} \sum_{n=1}^{N}  P'_{i,u,n}  y_{n,s-1},\label{eqn:average_load_4}
\end{align}
and $P'_{i,u,n}$ is given in \eqref{eqn:region1}-\eqref{eqn:region4}  at the top of the next page.
\end{Lem}
\begin{proof}
Please refer to Appendix A.
\end{proof}

\begin{figure*}[!t]
\small{
\begin{align}
P'_{i,u,n} &=\sum_{a \in \{1,\ldots, u-2\}}\sum_{b_1 \in \{a, \ldots, i+a-1\}}\sum_{b_2 \in \{u-a-1, \ldots, K-i-a-1\}}
{K \choose b_1,K-b_1-b_2, b_2} P_n^{K-b_1-b_2}  \nonumber \\
& \times \sum_{\mathcal L_1 \in \mathcal L_{a,1}}\sum_{(\alpha_{n'})_{n' \in \mathcal L_1} \in \mathcal A_{b_1, \mathcal L_1}} \frac{b_1!}{\prod_{n' \in \mathcal L_1}\alpha_{n'}!}\prod_{n' \in \mathcal L_1}P_{n'}^{\alpha_{n'}} \sum_{\mathcal L_2 \in \mathcal L_{a,2}}\sum_{(\alpha_{n'})_{n' \in \mathcal L_2} \in  \mathcal A_{b_2, \mathcal L_2}}  \frac{b_2!}{\prod_{n' \in \mathcal L_2}\alpha_{n'}!}\prod_{n' \in \mathcal L_2}P_{n'}^{\alpha_{n'}}\nonumber \\
& +\sum_{b_2 \in \{u-1, \ldots, K-2\}} {K \choose  b_2} P_n^{K-b_2} \sum_{\mathcal L_2 \in \mathcal L_{0,2}}\sum_{(\alpha_{n'})_{n' \in \mathcal L_2} \in  \mathcal A_{b_2, \mathcal L_2}} \frac{b_2!}{\prod_{n' \in \mathcal L_2}\alpha_{n'}!}\prod_{n' \in \mathcal L_2}P_{n'}^{\alpha_{n'}}\nonumber \\
& +\sum_{b_1 \in \{u-1, \ldots, K-2\}} {K \choose  b_1} P_n^{K-b_1} \sum_{\mathcal L_1 \in \mathcal L_{u-1,1}}\sum_{(\alpha_{n'})_{n' \in \mathcal L_1} \in  \mathcal A_{b_1, \mathcal L_1}}  \frac{b_1!}{\prod_{n' \in \mathcal L_1}\alpha_{n'}!}\prod_{n' \in \mathcal L_1}P_{n'}^{\alpha_{n'}}, u \leq n,  u+n \leq N+1 \label{eqn:region1}
\end{align}}
\small{
\begin{align}
P'_{i,u,n}&=\sum_{a \in \{u-1+n-N,\ldots, u-2\}}\sum_{b_1 \in \{a, \ldots, i+a-1\}}\sum_{b_2 \in \{u-a-1, \ldots, K-i-a-1\}} {K \choose b_1,K-b_1-b_2, b_2} P_n^{K-b_1-b_2}  \nonumber \\
 & \times \sum_{\mathcal L_1 \in \mathcal L_{a,1}}\sum_{(\alpha_{n'})_{n' \in \mathcal L_1} \in \mathcal A_{b_1, \mathcal L_1}} \frac{b_1!}{\prod_{n' \in \mathcal L_1}\alpha_{n'}!}\prod_{n' \in \mathcal L_1}P_{n'}^{\alpha_{n'}} \sum_{\mathcal L_2 \in \mathcal L_{a,2}}\sum_{(\alpha_{n'})_{n' \in \mathcal L_2} \in  \mathcal A_{b_2, \mathcal L_2}} \frac{b_2!}{\prod_{n' \in \mathcal L_2}\alpha_{n'}!}\prod_{n' \in \mathcal L_2}P_{n'}^{\alpha_{n'}}\nonumber \\
& +\sum_{b_1 \in \{u-1, \ldots, K-2\}} {K \choose  b_1} P_n^{K-b_1} \sum_{\mathcal L_1 \in \mathcal L_{u-1,1}}\sum_{(\alpha_{n'})_{n' \in \mathcal L_1} \in  \mathcal A_{b_1, \mathcal L_1}}   \frac{b_1!}{\prod_{n' \in \mathcal L_1}\alpha_{n'}!}\prod_{n' \in \mathcal L_1}P_{n'}^{\alpha_{n'}}, u \leq n, u+n > N+1\label{eqn:region2}
\end{align}}
\small{
\begin{align}
P'_{i,u,n}&=\sum_{a \in \{1,\ldots, n-1\}}\sum_{b_1 \in \{a, \ldots, i+a-1\}}\sum_{b_2 \in \{u-a-1, \ldots, K-i-a-1\}}  {K \choose b_1,K-b_1-b_2, b_2} P_n^{K-b_1-b_2} \nonumber \\
& \times \sum_{\mathcal L_1 \in \mathcal L_{a,1}}\sum_{(\alpha_{n'})_{n' \in \mathcal L_1} \in \mathcal A_{b_1, \mathcal L_1}}\frac{b_1!}{\prod_{n' \in \mathcal L_1}\alpha_{n'}!}\prod_{n' \in \mathcal L_1}P_{n'}^{\alpha_{n'}} \sum_{\mathcal L_2 \in \mathcal L_{a,2}}\sum_{(\alpha_{n'})_{n' \in \mathcal L_2} \in  \mathcal A_{b_2, \mathcal L_2}} \frac{b_2!}{\prod_{n' \in \mathcal L_2}\alpha_{n'}!}\prod_{n' \in \mathcal L_2}P_{n'}^{\alpha_{n'}}\nonumber \\
& +\sum_{b_2 \in \{u-1, \ldots, K-2\}} {K \choose  b_2} P_n^{K-b_2} \sum_{\mathcal L_2 \in \mathcal L_{0,2}}\sum_{(\alpha_{n'})_{n' \in \mathcal L_2} \in  \mathcal A_{b_2, \mathcal L_2}}  \frac{b_2!}{\prod_{n' \in \mathcal L_2}\alpha_{n'}!}\prod_{n' \in \mathcal L_2}P_{n'}^{\alpha_{n'}}, u > n, u+n \leq N+1\label{eqn:region3}
\end{align}}
\small{
\begin{align}
P'_{i,u,n} &=\sum_{a \in \{u-1+n-N,\ldots, n-1\}}\sum_{b_1 \in \{a, \ldots, i+a-1\}}\sum_{b_2 \in \{u-a-1, \ldots, K-i-a-1\}} {K \choose b_1,K-b_1-b_2, b_2} P_n^{K-b_1-b_2} \nonumber \\
& \times \sum_{\mathcal L_1
 \in \mathcal L_{a,1}}\sum_{(\alpha_{n'})_{n' \in \mathcal L_1} \in \mathcal A_{b_1, \mathcal L_1}} \frac{b_1!}{\prod_{n' \in \mathcal L_1}\alpha_{n'}!}\prod_{n' \in \mathcal L_1}P_{n'}^{\alpha_{n'}}  \nonumber \\
& \times \sum_{\mathcal L_2 \in \mathcal L_{a,2}}\sum_{(\alpha_{n'})_{n' \in \mathcal L_2} \in  \mathcal A_{b_2, \mathcal L_2}} \frac{b_2!}{\prod_{n' \in \mathcal L_2}\alpha_{n'}!}\prod_{n' \in \mathcal L_2}P_{n'}^{\alpha_{n'}}, u > n, u+n > N+1\label{eqn:region4}
\end{align}}\hrulefill
\end{figure*}

Problem~\ref{Prob:simplify_2} is a linear  program  with $N(K+1)$ variables  and  can be solved  by using linear optimization techniques.

\subsection{Optimization  for Uniform  File   Popularity}\label{Sub:uniform}
In this part, we consider a special case, i.e., the uniform file popularity  ($p_n =\frac{1}{N}$,  for all  $n \in \mathcal N$).  First, we  present another   structural condition on the file partition parameter.

\begin{Con}[Symmetry w.r.t. File]\label{Con:popularity_2}
For all  $n \in \{1,2,\ldots, N-1\}$ and $s \in \{1,2, \cdots, K\}$,  when $p_n=p_{n+1}$,
\begin{align}
y_{n,s} = y_{n+1,s}. \label{eqn:symmetry_2}
\end{align}
\end{Con}

Condition~\ref{Con:popularity_2} indicates that for all   $n \in \{1,2,\ldots, N-1\}$ and   $s \in \{1,2, \cdots, K\}$,  when $p_n=p_{n+1}$,
the size of subfiles $W_{n,\mathcal{S}}$, $\mathcal{S}\subseteq \{\mathcal{\widehat{S}} \subseteq \mathcal{K}: |\mathcal{\widehat{S}}|=s\}$ is the same as that of subfiles
$W_{n+1,\mathcal{S}}$, $\mathcal{S}\subseteq \{\mathcal{\widehat{S}} \subseteq \mathcal{K}: |\mathcal{\widehat{S}}|=s\}$.
Condition~\ref{Con:popularity_2}  ensures
zero  variance of the lengths of messages involved in the coded-multicast XOR operations,  hence  avoiding ``bit waste" effect and  further  increasing coded-multicasting opportunities  for the uniform file popularity.
Later, we shall show that imposing this condition will not lose optimality of Problem~\ref{Prob:simplify_2}  under the uniform file popularity.

By Condition~\ref{Con:popularity_2}, we can set
\begin{align}
y_{n,s}=z_s, \quad  \forall  s \in \{0,1,\cdots, K\}, \ n \in \mathcal{N}. \label{eqn:symmetry_2}
\end{align}
Here,   $z_s$  can be viewed as  the size of each subfile of type $s$, normalized by the file size $F$.
Let $\mathbf{z} \triangleq (z_{s})_{ s \in \{0,1,\cdots, K\}}$.

Next, we simplify Problem~\ref{Prob:original} under Conditions~\ref{Con:symmetry} and  \ref{Con:popularity_2}.
\textcolor{black}{First, we introduce some notations.}
Let $P''_{u} \triangleq \Pr \left[|\underline{\mathcal{K}}(\mathbf D)|=u\right]$,   for all  $u \in \left\{1,2, \ldots, \min\{K,N\}\right\}$.
Note that event $|\underline{\mathcal{K}}(\mathbf D)|=u$ corresponds to   the event  in the  ``balls into bins'' problem that there are $u$ nonempty bins after placing  $K$ balls uniformly at random into $N$ bins.
By using  results  for the   ``balls into bins'' problem  in the uniform case~\cite{flajolet2009analytic} and considering Conditions~\ref{Con:symmetry} and~\ref{Con:popularity_2},  we have the following result.
\begin{Lem}[Simplification \textcolor{black}{of Problem~\ref{Prob:original} for}  Uniform File Popularity] \label{Lem:uniform_obj_simp}
\textcolor{black}{Under   Conditions~\ref{Con:symmetry} and~\ref{Con:popularity_2},
Problem~\ref{Prob:original} can be  converted into:}
\begin{Prob}[Simplified  Problem   for Uniform File  Popularity]\label{Prob:equivalent_3}
\begin{align}
\widehat{R}^*_{\rm avg}(K,N,M)\triangleq \min_{\mathbf{z}} \quad &   \sum_{s=0}^{K-1}{K \choose s+1}z_{s}- \sum_{u=1}^{\min \{K,N\}} P''_{u}
\sum_{s=0}^{K-u-1}{K-u \choose s+1} z_{s}\nonumber \\
s.t. \quad  &0 \leq z_{s} \leq 1 , \quad s \in \{0,1,\cdots, K\}, \label{eqn:X_range_3}\\
&\sum_{s=0}^{K} {K \choose s} z_{s}=1 , \label{eqn:X_sum_3}\\
&\sum_{s=0}^{K} {K \choose s}s z_{s} \leq \frac{KM}{N}, \label{eqn:memory_constraint_3}
\end{align}
where    
$
P''_{u}=\begin{Bmatrix}
K\\
u
\end{Bmatrix}
\frac{{N \choose u}u!}{N^K},$
and
$\begin{Bmatrix}
K\\
u
\end{Bmatrix}$
is the  Stirling number of the second kind.
\end{Prob}
\end{Lem}
\begin{proof}
Please refer to Appendix B.
\end{proof}

Problem~\ref{Prob:equivalent_3} is a linear  program  with  $K+1$   variables and  can be solved more efficiently than Problem~\ref{Prob:original}.

Finally,  we discuss the relation between an optimal solution of Problem~\ref{Prob:equivalent_3} and  Yu {\em et al.}'s centralized coded caching  scheme~\cite{YuQian}.\footnote{Yu {\em et al.}'s centralized coded caching  scheme focuses on  cache size $M \in \{0, \frac{N}{K}, \frac{2N}{K}, \ldots, N\}$, so that $\frac{KM}{N}$ is an integer in  $\{0, 1, \ldots, K\}$.  For general $M \in [0,N]$, the worst-case load can be achieved by memory sharing.}
Using KKT conditions, we have
\begin{Lem}[Optimal Solution to Problem~\ref{Prob:equivalent_3}]\label{Lem:Ali}
For cache size $M \in \left\{0, \frac{N}{K}, \frac{2N}{K}, \ldots, N\right\}$, $\mathbf{z}^*\triangleq (z^*_{s})_{ s \in \{0,1,\cdots, K\}}$ is an   optimal solution  to  Problem~\ref{Prob:equivalent_3}, where
\begin{align}
z^*_{s}=
\begin{cases}
\frac{1}{{K \choose \frac{KM}{N}}}, & s=\frac{KM}{N}\\
0, &s \in \{0,1,\cdots, K\}\setminus \{\frac{KM}{N}\},
\end{cases}\label{eqn:uniform_caching}
\end{align}
and the  optimal value of Problem~\ref{Prob:equivalent_3} is given by\footnote{In this paper, we define ${n \choose k}=0$ when $k>n$~\cite{YuQian}.}
\begin{align}
\widehat{R}^*_{\rm avg}(K,N, M)&= \frac{K(1-M/N)}{1+KM/N}- \sum_{u=1}^{\min \{K,N\}}
P''_{u}{K-u \choose  \frac{KM}{N}+1}\Big /{K \choose \frac{K M}{N}}.  \label{eqn:Ali}
\end{align}
\end{Lem}
\begin{proof}
Please refer to Appendix C.
\end{proof}

Lemma~\ref{Lem:Ali} indicates that  Yu {\em et al.}'s centralized coded caching  scheme corresponds to an optimal solution of Problem~\ref{Prob:equivalent_3}. In addition, the optimal average load $\widehat{R}^*_{\rm avg}(K,N, M)$  in~\eqref{eqn:Ali}  is equivalent to    that in~\cite{YuQian}.
Specifically, the first term in~\eqref{eqn:Ali} corresponds to the  worst-case load in~\cite{AliFundamental} and the second term  in~\eqref{eqn:Ali}, which is given in explicit form as opposed to the implicit form containing an expectation w.r.t. the random  requests  given   in~\cite{YuQian},  indicates the load reduction due to the ability of the delivery strategy to take advantage of common requests.
Note that it has been shown that Yu {\em et al.}'s centralized coded caching  scheme
is optimal among all uncoded placement and all delivery  under the uniform file   popularity.
Thus, 
for the uniform file popularity, Conditions 1  and 3 are actually optimal properties.

\section{Average Load Minimization with Subpacketization Level Constraint} \label{Sec:L0}
In this section, we minimize the  average load  by optimizing  the file partition parameter  under   the subpacketization  level  constraint for each file, which  is  given by
\begin{align}
\| \mathbf{x}_n \|_{0} \leq \widehat{F}, \ n \in \mathcal N, \label{eqn:subpack}
\end{align}
where  $\| \mathbf{x}_n \|_{0}  \triangleq  \sum_{\mathcal S \in \mathcal K}\mathbf{1} \left[x_{n,\mathcal{S}}\neq0\right]  \in \{1,2, \ldots,  2^K\}$  denotes the $\ell_0$-norm of  the vector  $\mathbf{x}_n$, i.e., the total number of subfiles for  file $W_n$, and   $\widehat{F} \in \{1,2, \ldots,  2^K\}$ represents the  maximum admissible subpacketization level   for all files.
To the best of our knowledge, this is the first work  explicitly considering  the subpacketization  level  constraint in optimizing  coded caching design.

\subsection{Problem Formulation}\label{Sub:formulation_L0}
In this part,  we  minimize the average load under the file partition constraints in  \eqref{eqn:X_range} and  \eqref{eqn:X_sum}, the cache memory constraint in \eqref{eqn:memory_constraint}, and the  subpacketization  level  constraint in~\eqref{eqn:subpack}.
\begin{Prob}[Optimization  for Arbitrary File Popularity  with Subpacketization Constraint]\label{Prob:subpack1}
\begin{align}
R^{\dag}_{\rm avg}(K,N,M) \triangleq \min_{\mathbf{x}} \quad &  R_{\rm avg}(K,N,M,\mathbf{x})\nonumber \\
s.t. \quad  &\eqref{eqn:X_range}, \eqref{eqn:X_sum}, \eqref{eqn:memory_constraint}, \eqref{eqn:subpack},\nonumber
\end{align}
where $R_{\rm avg}(K,N,M,\mathbf{x})$ is given by \eqref{eqn:average_load_1}.
\end{Prob}

Compared with  Problem~\ref{Prob:original},  Problem~\ref{Prob:subpack1} has  an extra constraint, i.e., the  subpacketization  level  constraint in~\eqref{eqn:subpack}.
There are two main challenges in solving  Problem~\ref{Prob:subpack1}.
First, Problem~\ref{Prob:subpack1} is  an  NP-Hard   problem  due to the   combinatorial  constraint in~\eqref{eqn:subpack} involving  the $\ell_0$-norm~\cite{dc2017}.
Second,   as in Problem~\ref{Prob:original}, the number of variables in Problem~\ref{Prob:subpack1} is $N 2^K$,  which  is huge,  especially when $K$ and $N$ are large.

\subsection{Simplified Formulation}
There are  extensive  research  dealing with optimization problems involving the $\ell_0$-norm.
Those  works can be divided into three  main  categories according to the way  of treating   the $\ell_0$-norm,  i.e.,  convex approximation, non-convex approximation, and non-convex exact reformulation~\cite{le2015dc}.
For the  category  of convex approximation, one of the best known approaches is approximating  the  $\ell_0$-norm with the $\ell_1$-norm~\cite{le2015dc}.
If the original optimization problem is convex except the      constraint   involving the $\ell_0$-norm, this convex approximation  approach can  transform the original  NP-Hard  problem   into a convex  problem.
However,   it has been shown that  an optimal   solution of the  approximated convex   problem is not always sparse  (may not  be a feasible solution of the original problem)~\cite{shalev2010trading}.\footnote{Given the constraint in \eqref{eqn:X_sum_2}, the $\ell_1$-norm of $\mathbf{x}_n$  is  equal to a constant, i.e., $\| \mathbf{x}_n \|_{1}= \sum_{i=1}^{2^K}x_i=1$.
Thus, replacing  $\| \mathbf{x}_n \|_{0}$ with   $\| \mathbf{x}_n \|_{1}$ in \eqref{eqn:subpack} results in the constraint $1 \leq \widehat{F}$, which always holds and cannot restrain $\mathbf{x}_n$.}
For the category of non-convex approximation,  a variety of sparsity-inducing penalty functions, e.g., the $\ell_p$ pseudo-norm with $0<p<1$~\cite{lp}, exponential concave function~\cite{bradley1998feature}, and logarithmic function~\cite{Weston2003},  have been proposed to approximate the $\ell_0$-norm.
In general, non-convex approximation can provide better sparsity than convex  approximation, but may  still  not provide a feasible  solution of the original problem.
Few works focus on non-convex exact reformulation, which is proposed to guarantee the equivalence  between the reformulated problem and the original problem.
Using exact penalty techniques, \cite{le2015dc} and \cite{Advances2013}   show that   the  reformulated  problems  with suitable parameters  are equivalent to the original problems (share the same feasible solutions with  the original problems).
However, the  reformulated problems  are quite convoluted  as they rely on  several  parameters~\cite{dc2017}.
In  the  recent work~\cite{dc2017}, the authors propose an  exact DC    reformulation which is  simpler than the reformulated problems proposed in \cite{le2015dc} and \cite{Advances2013}, and then  obtain a stationary point of the DC problem using a DC algorithm.
In the following,    we use the  exact DC reformulation method  in~\cite{dc2017} in order to obtain a simple equivalent formulation of the original problem.

We  first  simplify Problem~\ref{Prob:subpack1} to   facilitate a  low-complexity solution.
Let  $a_{[i]}$ denote the  element whose  value is the $i$-th largest among the $m$ elements of  the vector  $\mathbf{a}$, i.e., $a_{[1]} \geq a_{[2]} \geq \ldots \geq a_{[m]} $.
Let $\| \mathbf{a} \|_{lgst, \widehat{F}}$ denote the largest-$\widehat{F}$ norm  of  the vector  $\mathbf{a}$, i.e.,
$\| \mathbf{a} \|_{lgst, \widehat{F}} \triangleq |a_{[1]}|+|a_{[2]}|+ \ldots +|a_{[\widehat{F}]}|$~\cite{cvx}.
Using the method for obtaining Problem~\ref{Prob:simplify_2} and  Theorem~1 of~\cite{dc2017} for simplifying the constraint in~\eqref{eqn:subpack} under Conditions~\ref{Con:symmetry} and~\ref{Con:popularity}, we have the following result.

\begin{Lem}[Simplification for  \textcolor{black}{Problem~\ref{Prob:subpack1} for Arbitrary File Popularity}]\label{Lem:SimpSubpackCons}
Under   Conditions~\ref{Con:symmetry} and~\ref{Con:popularity}, Problem~\ref{Prob:subpack1} can be  converted into:
\begin{Prob}[Simplified Problem   for Arbitrary File Popularity] \label{Prob:subpack2}
\begin{align}
\widetilde{R}^\dag_{\rm avg}(K,N,M,\mathbf{y}) \triangleq  \min_{\mathbf{y}} \quad &  \widetilde{R}_{\rm avg}(K,N,M,\mathbf{y})\nonumber \\
s.t. \quad  &\eqref{eqn:additional_cons}, \eqref{eqn:X_range_2}, \eqref{eqn:X_sum_2}, \eqref{eqn:memory_constraint_2}, \nonumber \\
& \|U_n \mathbf{y} \|_{lgst, \widehat{F}} \geq 1, \ n \in \mathcal N,  \label{eqn:subpack_eq_y}
\end{align}
where $\widetilde{R}_{\rm avg}(K,N,M,\mathbf{y})$ is given by \eqref{eqn:average_load_4},
\begin{align}
U_n \triangleq
[\mathbf{0},   \cdots, \mathbf{0}, \underbrace{J}_{n\text{-th  block}}, \mathbf{0}, \cdots, \mathbf{0}]
 \label{eqn:W}
\end{align}
denotes the   block matrix of the   dimension $2^K\times (K+1)N$, with   $2^K\times(K+1)$   matrix   $J \triangleq (j_{m,l})_{m \in \{1, \ldots, 2^K\}, l \in \{1, \ldots, K+1\}}$ as its  $n$-th block  and  $2^K\times(K+1)$   zero matrices as other blocks,  and the element of row $m$  and column $l$  of  $J$ is
\begin{align}\label{eqn:J}
j_{m,l} =\begin{cases}
1, & \sum_{i=1}^{l-1} {K \choose i-1}<m \leq \sum_{i=1}^{l} {K \choose i-1}\\
0, &\text{otherwise}
\end{cases}.
\end{align}
\end{Prob}
\end{Lem}
\begin{proof}
Please refer to Appendix D.
\end{proof}

The number of variables in Problem~\ref{Prob:subpack2} is $N(K+1)$, which is much smaller than that of Problem~\ref{Prob:subpack1},  i.e., $N2^K$.
In addition, compared with Problem~\ref{Prob:simplify_2},  Problem~\ref{Prob:subpack2} has  an extra constraint in~\eqref{eqn:subpack_eq_y}, which has two advantages over the  subpacketization  level  constraint in~\eqref{eqn:subpack}:
(i)  the constraint in~\eqref{eqn:subpack_eq_y}  is a   DC constraint,  making Problem~\ref{Prob:simplify_2} a DC   problem, which can be solved by  a   DC algorithm in polynomial time;
(ii)   a  subgradient of  $\| U_n \mathbf{y} \|_{lgst, \widehat{F}}$  can be efficiently computed,  making  the DC algorithm  an efficient one.

Thus, in the following, we solve Problem~\ref{Prob:subpack2} by using  a  DC algorithm.
The main idea of the DC algorithm is to  iteratively solve a sequence of convex problems,  each of which is obtained by  linearizing  the second term of the objective function of the DC problem.
A subgradient of the second  term of the objective function is required in the linearization in each iteration.
Thus, to solve  Problem~\ref{Prob:subpack2}, we first  obtain  a subgradient of $\|U_n \mathbf{y} \|_{lgst, \widehat{F}}$ by extending   the  closed-form expression of  a   subgradient of  $\|\mathbf{y} \|_{lgst, \widehat{F}}$   given in~\cite{dc2017}.
\begin{Lem}[Subgradient of  $\|U_n\mathbf{y} \|_{lgst, \widehat{F}}$]\label{Lem:subgra}
$\mathbf{g}_n(\mathbf y) \triangleq (g_{n,i}(\mathbf y))_{i \in \{1,2, \ldots, N(K+1)\}}$  is a subgradient of  $\|U_n \mathbf{y} \|_{lgst, \widehat{F}}$, where
\begin{align}
&g_{n,(m-1)(K+1)+[i]}(\mathbf y)=\begin{cases}
{K \choose [i]-1}, & m=n, \ i \in \{1, \ldots, I-1\} \\
\widehat{F}- \sum_{i=1}^{I-1} {K \choose [i]-1}, & m=n, \  i=I \\
0, &\text{otherwise}
\end{cases},\label{eqn:g_t}
\end{align}
$[i]$ represents the index of $i$-th largest element in  $\mathbf{y}_n$,   and $I$ satisfies  $\sum_{i=1}^{I-1} {K \choose  [i]-1  }   \leq \widehat{F}$ and $\sum_{i=1}^{I} {K \choose  [i]-1  }   > \widehat{F}$.
\end{Lem}
\begin{proof}
Please refer to Appendix E.
\end{proof}

Then, based on  the subgradient   $\mathbf{g}_n(\mathbf y)$  given in   Lemma~\ref{Lem:subgra}, we can obtain a stationary point of Problem~\ref{Prob:subpack2} using the DC algorithm as   summarized   in Alg.~\ref{alg:concave}~\cite{Boyd2016}.
As in~\cite{le2015dc},  to  approach   a globally optimal solution of Problem~\ref{Prob:subpack2},  we   obtain multiple   stationary points of Problem~\ref{Prob:subpack2} by  performing the DC algorithm  multiple times,  each with  a  random initial feasible point  of Problem~\ref{Prob:subpack2},    and adopt the   stationary point with the lowest average load among all the obtained   stationary points  of Problem~\ref{Prob:subpack2}.

\begin{algorithm}[h]
\caption{DC  Algorithm for Solving Problem~\ref{Prob:subpack2}}
\begin{algorithmic}[1]
\small{\STATE Find an initial feasible point $\mathbf y^{(0)}$   of Problem~\ref{Prob:subpack2} and set  $t=0$
\STATE  \textbf{repeat}
\STATE Set   $\mathbf y^{(t+1)}$ to  be  an optimal solution of the convex problem:
\begin{align}
 \min_{\mathbf{y}} \quad &  \widetilde{R}_{\rm avg}(K,N,M,\mathbf{y})\nonumber \\
s.t. \quad  &\eqref{eqn:additional_cons}, \eqref{eqn:X_range_2}, \eqref{eqn:X_sum_2}, \eqref{eqn:memory_constraint_2}  \nonumber\\
&  \|U_n\mathbf{y}^{(t)} \|_{lgst, \widehat{F}}+\mathbf{g}_n(\mathbf y^{(t)})^T(\mathbf y -\mathbf y^{(t)}) \geq 1, \ n \in \mathcal N
\end{align}
\STATE   Set  $t=t+1$
\STATE  \textbf{until}   $\widetilde{R}_{\rm avg}(K,N,M,\mathbf{y}^{(t-1)})- \widetilde{R}_{\rm avg}(K,N,M,\mathbf{y}^{(t)})  \leq \delta$}
\end{algorithmic}\label{alg:concave}
\end{algorithm}

\section{Numerical Results}\label{Sec:Numerical}
\begin{figure}
\begin{center}
  \subfigure[\small{Case without considering the subpacketization  level  issue.}]
  {\resizebox{5.2cm}{!}{\includegraphics{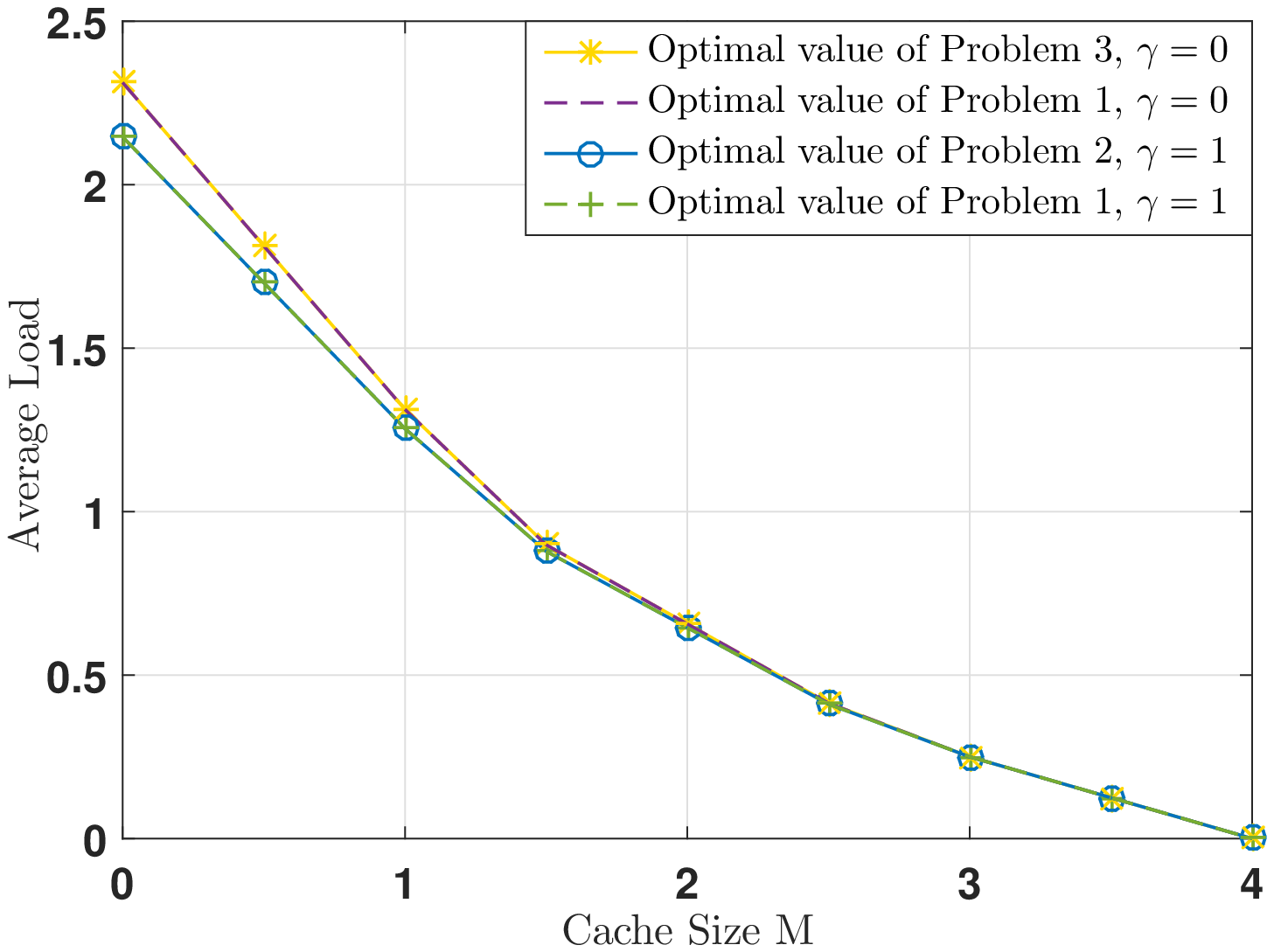}}}
      \subfigure[\small{Case  considering the subpacketization  level  issue at $\widehat{F}=5$.}]
  {\resizebox{5.2cm}{!}{\includegraphics{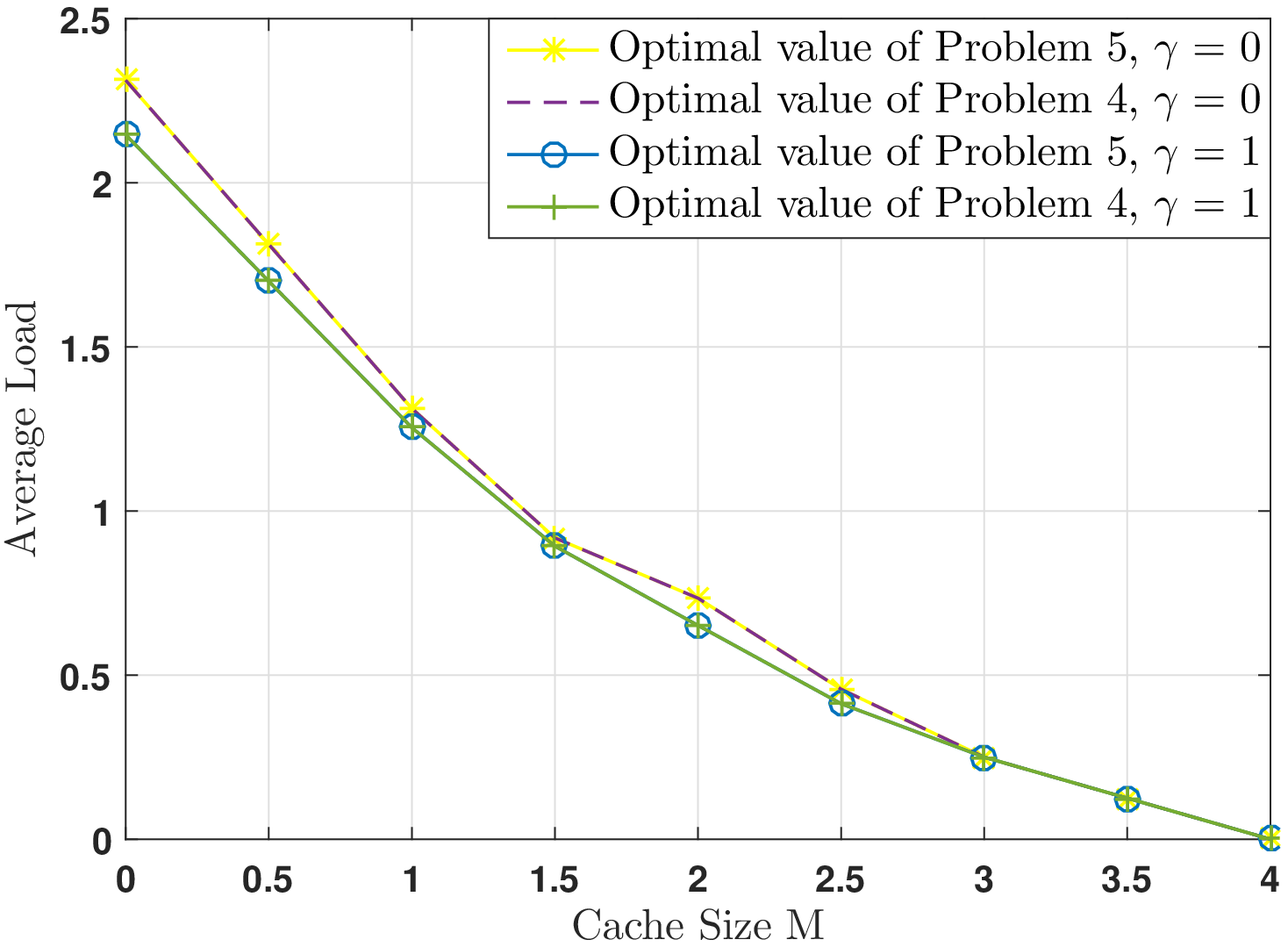}}}
  \subfigure[\small{Case  considering the subpacketization  level  issue at $M=2$.}]
  {\resizebox{5.2cm}{!}{\includegraphics{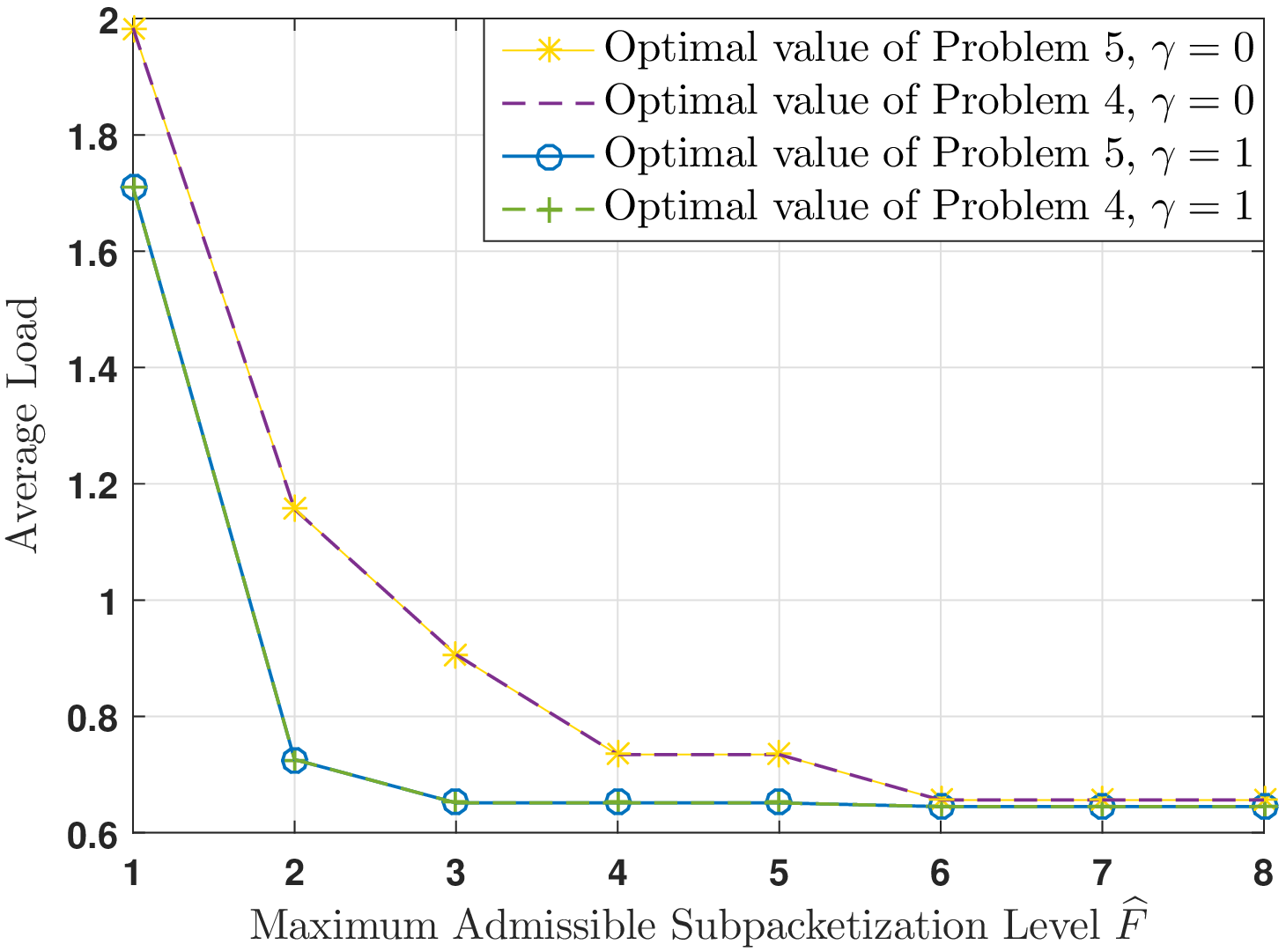}}}
  \end{center}
         \caption{Verification of Conditions 1, 2 and 3 in both cases at  $K=3$ and $N=4$.
         }\label{fig:verify}
\end{figure}

In the simulation, we assume  that  the file popularity follows Zipf distribution, i.e., $p_n=\frac{n^{-\gamma}}{\sum_{n \in \mathcal N} n^{-\gamma}}$ for all $n\in \mathcal N$, where $\gamma$ is the Zipf exponent.
Fig.~\ref{fig:verify}~(a) shows the optimal values  of  Problems~\ref{Prob:original}, \ref{Prob:simplify_2}  and \ref{Prob:equivalent_3},  verifying  that Conditions 1, 2 and 3 are  optimal conditions  in  the case without considering
the subpacketization  level  issue.
Fig.~\ref{fig:verify}~(b) and Fig.~\ref{fig:verify}~(c) show  the optimal values  of  Problems~\ref{Prob:subpack1} and~\ref{Prob:subpack2}, verifying  that Conditions 1   and 2 are  optimal conditions in  the case  considering the subpacketization  level  issue.

Fig.~\ref{fig:compare}  compares the average load of our optimized parameter-based scheme, the average loads of  Maddah-Ali--Niesen's centralized scheme~\cite{AliFundamental}, Jin {\em et al.}'s centralized scheme~\cite{jin},  Yu {\em et al.}'s centralized scheme~\cite{YuQian},   the genie-aided converse bound in~\cite{jin} and the conserve bound in~\cite{bound16},   all without considering  the subpacketization  level  issue.
From Fig.~\ref{fig:compare}, we can see that the optimized parameter-based scheme   outperforms  the three baseline  schemes.
The gain    over  Jin {\em et al.}'s optimized centralized coded caching scheme  follows by using an extended version of the improved delivery strategy of Yu {\em et al.}, that takes advantage of common requests (which occur with positive probability in the case of random requests).
The gain     over Yu {\em et al.}'s centralized coded caching   scheme is due to   exploiting  the explicit knowledge of the file popularity in the optimization of  content placement.\footnote{It  has been  proved in~\cite{jin} that  the   optimized  parameter-based scheme in~\cite{jin} outperforms Maddah-Ali--Niesen's centralized coded caching scheme~\cite{AliFundamental}.}
In addition, the optimized average load  is close to the converse bounds,  implying that  the optimal value obtained by solving  Problem~\ref{Prob:simplify_2} is close to optimal.

\begin{figure}
\begin{center}
  {\resizebox{7cm}{!}{\includegraphics{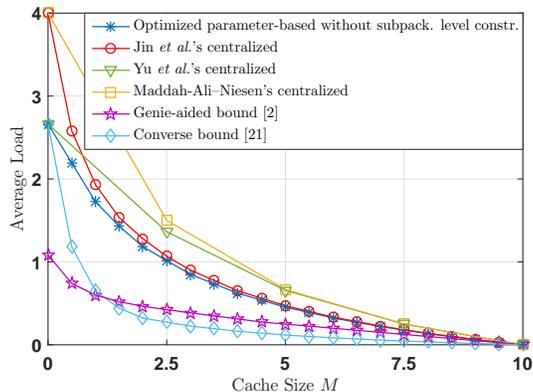}}}
         \caption{Average load versus $M$   in the case without considering  the subpacketization  level  issue  at $K=4$, $N=10$ and $\gamma=1.5$. Note that    Maddah-Ali--Niesen's and Yu {\em et al.}'s centralized coded caching schemes mainly focus on the cache size $M \in \left\{0,\frac{N}{K}, \frac{2N}{K}, \ldots, N\right\}$.  For other  $M \in [0, N]$, the average loads of Maddah-Ali--Niesen's and Yu {\em et al.}'s centralized coded caching schemes are achieved by memory sharing~\cite{AliFundamental,YuQian}.
         }\label{fig:compare}
\end{center}
\end{figure}

Fig.~\ref{fig:compare_subpack}  compares  the  average load  of  our optimized parameter-based scheme  considering    the subpacketization level  constraint, the  average loads  of Tang {\em et al.}'s scheme~\cite{Tang} and  Pareto-optimal PDA~\cite{cheng}  both considering  the subpacketization level  issue.
From Fig.~\ref{fig:compare_subpack}, we see that  the average load  of our optimized parameter-based scheme without considering  the  subpacketization  level  constraint  serves as a lower bound of the  proposed  one considering the  subpacketization  level  constraint.
Note that in Fig.~\ref{fig:compare_subpack}~(a),  for  the considered $K, N, \widehat{F}$, Tang {\em et al.}'s scheme is  applicable only at $M \in \{0, 2.5, 5, 7.5, 10\}$  and Pareto-optimal PDA is   applicable only at  $M \in \{2.5, 5, 7.5, 10\}$;
in Fig.~\ref{fig:compare_subpack}~(b),  for  the considered $K, N, M$, Tang {\em et al.}'s scheme is  applicable only at $\widehat{F}=4$  and Pareto-optimal PDA is   applicable only at  $\widehat{F}\in \{4, 20\}$;
in Fig.~\ref{fig:compare_subpack}~(c),  for   the considered $K, N, M$, Tang {\em et al.}'s scheme and Pareto-optimal PDA are not applicable at  any $\widehat{F}$.
In addition, from  Fig.~\ref{fig:compare_subpack}, we can see that our optimized parameter-based scheme outperforms Tang {\em et al.}'s scheme and  Pareto-optimal PDA in terms of  both  the  average load and  application  region.
From  Fig.~\ref{fig:compare_subpack}, we can see that our optimized parameter-based scheme  considering  the   subpacketization  level  constraint can achieve significantly  lower   subpacketization  level  than the one without considering  the subpacketization level  constraint,  at the cost of   a small  increase  of  the average load.
This means that sacrificing  a  small average load gain can achieve a  huge subpacketization  level  reduction under the considered setting.

\begin{figure}
\begin{center}
    \subfigure[\small{Average load versus $M$ at $K=4$, $N=10$, $\widehat{F}=4$ and $\gamma=1$.}]
  {\resizebox{5cm}{!}{\includegraphics{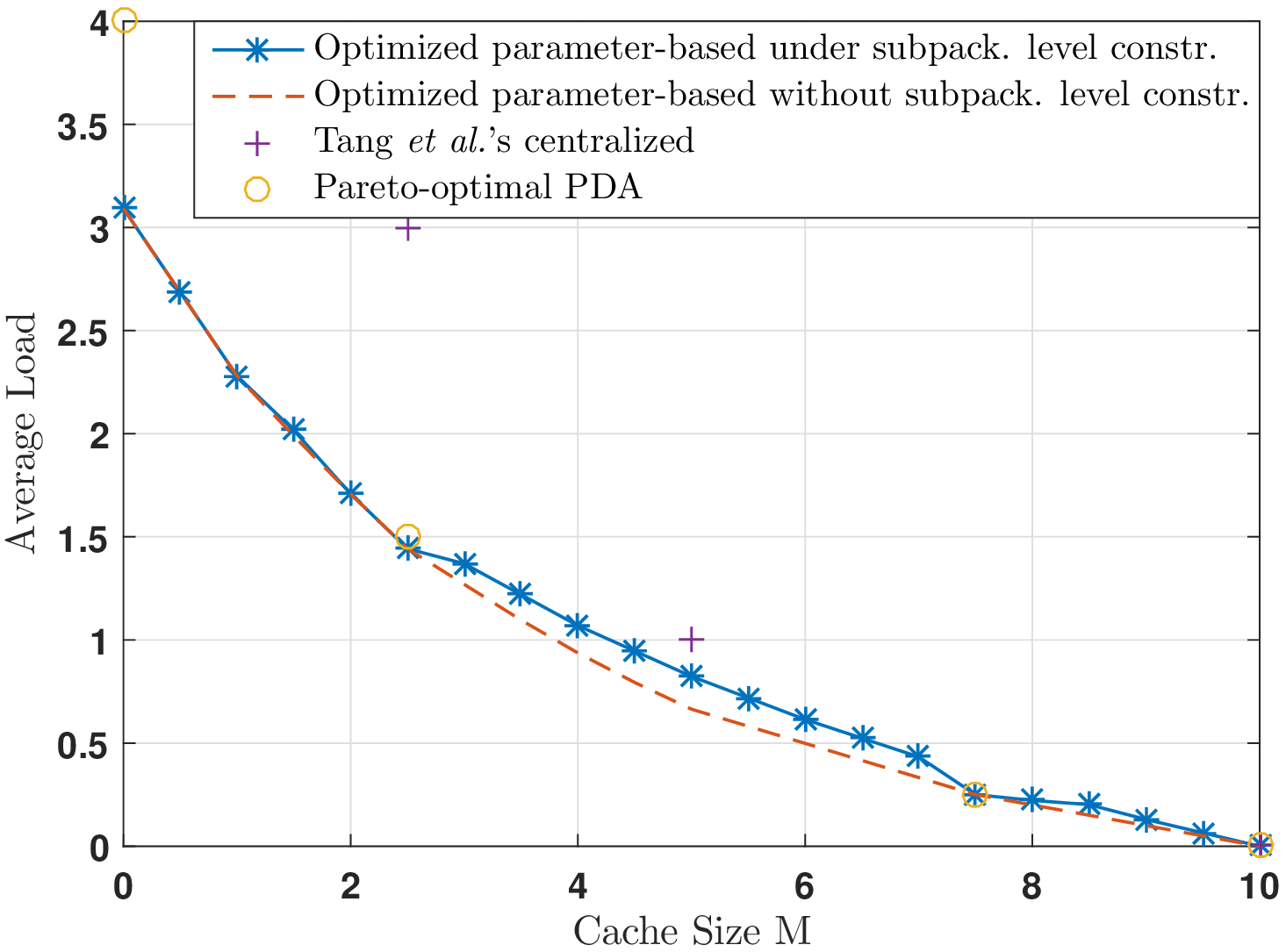}}}
  \subfigure[\small{Average load versus $\widehat{F}$ at  $K=6$, $N=5$, $M=2.5$, and $\gamma=1.4$.}]
  {\resizebox{5cm}{!}{\includegraphics{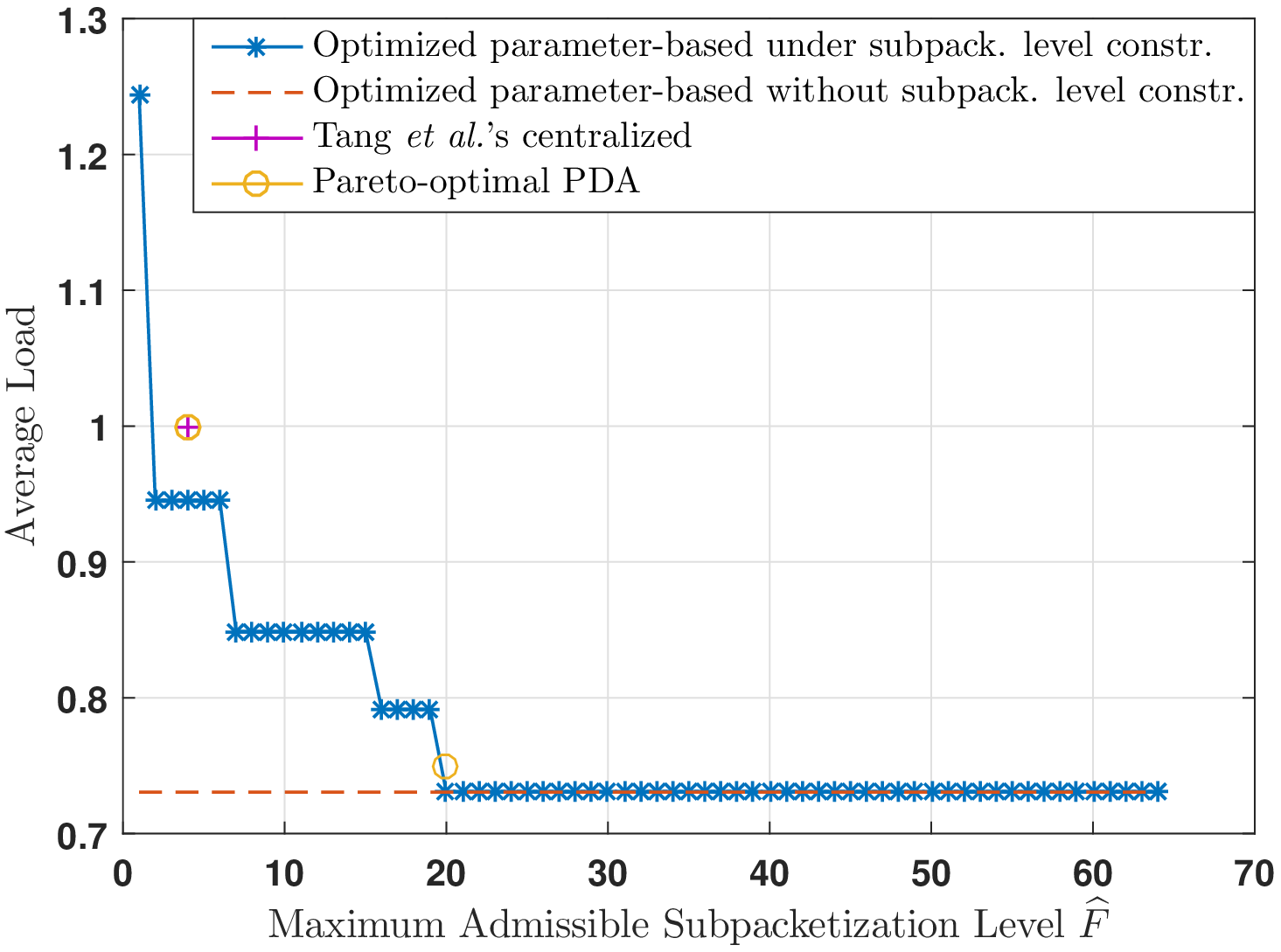}}}
    \subfigure[\small{Average load versus $\widehat{F}$ at  $K=6$, $N=10$, $M=4$, and  $\gamma=0.5$.}]
  {\resizebox{5cm}{!}{\includegraphics{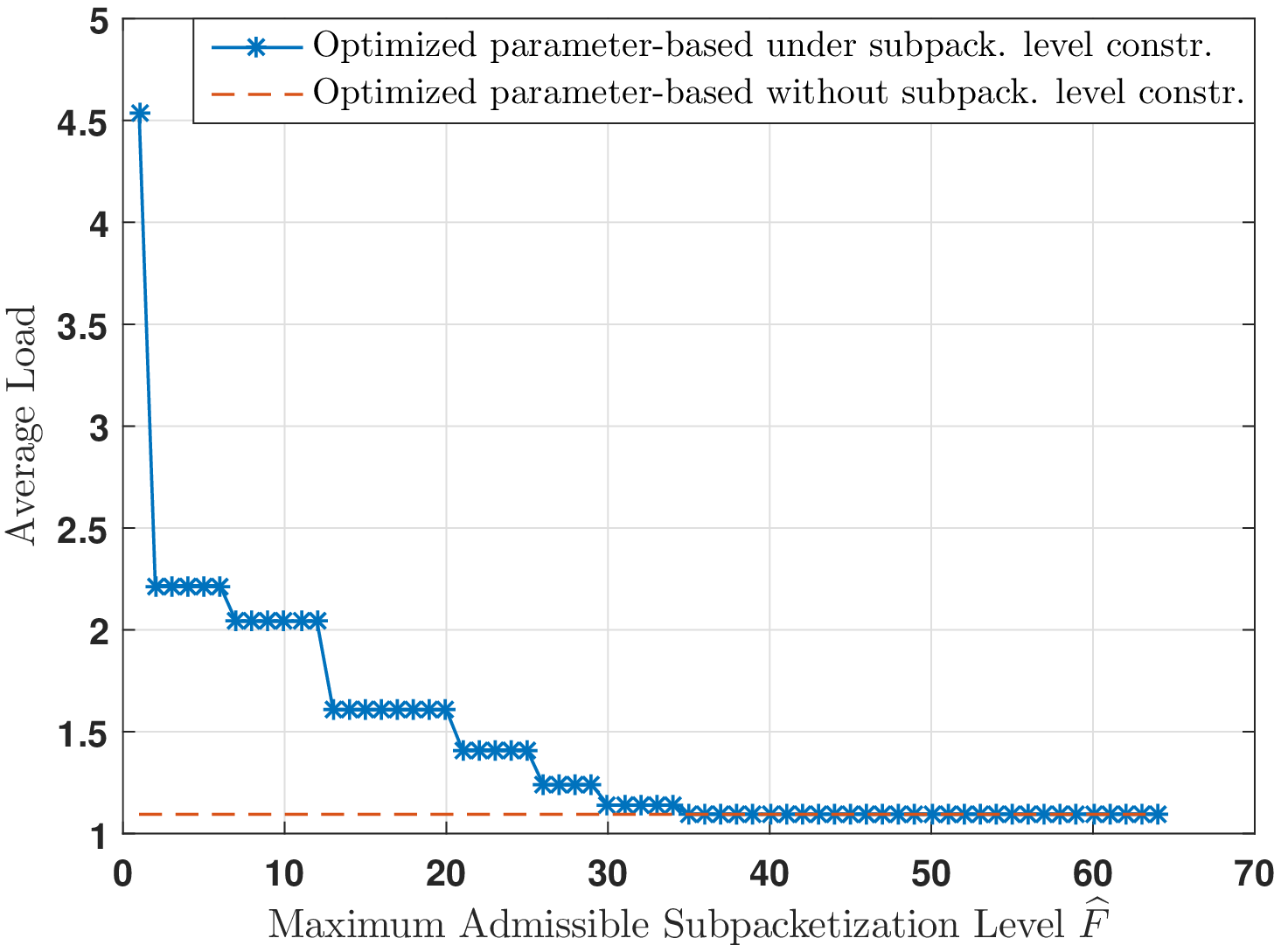}}}
  \end{center}
         \caption{\small{Average load versus  $M$ and  $\widehat{F}$  in the case considering the subpacketization level  issue.
         To obtain the average load of our optimized parameter-based scheme  under   the subpacketization  level  constraint,   we  solve Problem~\ref{Prob:subpack2} by performing    Alg.~\ref{alg:concave}   with $\delta=0.0001$ $100$ times each with a random initial feasible point and
         adopt the locally optimal solution with the lowest average load among all the obtained locally optimal solutions  of Problem~\ref{Prob:subpack2}  that are also feasible solutions of  Problem~\ref{Prob:subpack2}.}
         }\label{fig:compare_subpack}
\end{figure}

\section{Conclusion}
In this paper, we first presented a class of centralized coded caching schemes consisting of a general content placement strategy specified by a file partition parameter, enabling efficient and flexible content placement, and a specific content delivery strategy, enabling load reduction by exploiting common requests of different users.
Then we considered two cases: the case without considering the subpacketization  level  issue and the case considering the subpacketization level  issue.
In the first case, we formulated the coded caching optimization problem over the considered class of schemes  to minimize the average load under an arbitrary file popularity.
Imposing some conditions on the file partition parameter, we transformed the original optimization problem with $N2^K$ variables into a linear program with $N(K + 1)$ variables under an arbitrary file popularity and a linear program with $K+1$ variables under the uniform file popularity.
In the second case, we  formulated the coded caching optimization problem over the considered class of schemes to minimize the average load under an arbitrary file popularity  subject to    subpacketization  level  constraints involving  the $\ell_0$-norm.
Imposing the same conditions  and using the exact DC reformulation method, we converted the original problem with $N2^K$ variables into a simplified DC problem with $N(K + 1)$ variables,  which is solved using   DC  algorithm.
Finally, numerical results verify the optimality of the imposed conditions and demonstrate the advantages of the optimized scheme over existing schemes in both cases.

\section*{Appendix A: Proof of Lemma~\ref{Lem:simplification pop}}
By \eqref{eqn:symmetry}, it can be easily shown that the constraints in \eqref{eqn:X_range}, \eqref{eqn:X_sum} and \eqref{eqn:memory_constraint} of Problem~\ref{Prob:original} can be
converted into \eqref{eqn:X_range_2}, \eqref{eqn:X_sum_2} and \eqref{eqn:memory_constraint_2}. \textcolor{black}{Now, we}  show that the objective
function of Problem~\ref{Prob:original} in \eqref{eqn:average_load_1} can be converted into the objective function of Problem~\ref{Prob:simplify_2} in \eqref{eqn:average_load_4}.
First, by \eqref{eqn:symmetry}, we have
\begin{align}
&R_{\rm avg}(K,N,M,\mathbf{x})
\overset{(a)}=\sum_{\mathbf{d} \in \mathcal{N}^K} \left(\prod_{k=1}^{K}p_{d_k}\right) \sum_{s=1}^{K}\left(\sum_{\mathcal{S}   \subseteq \mathcal{K}: |\mathcal{S}|=s } \max_{k \in \mathcal{S}}y_{d_k,s-1} -\sum_{\mathcal{S}   \subseteq \mathcal{K} \setminus \underline{\mathcal{K}}(\mathbf d): |\mathcal{S}|=s } \max_{k \in \mathcal{S}}y_{d_k,s-1}\right) \nonumber\\
\overset{(b)}=&\sum_{s=1}^{K}{K \choose s}\sum_{n=1}^{N}\left(\left(\sum_{n'=n}^{N}p_{n'}\right)^s-\left(\sum_{n'=n+1}^{N}p_{n'}\right)^s\right) y_{n,s-1}
-\sum_{\mathbf{d} \in \mathcal{N}^K} \left(\prod_{k=1}^{K}p_{d_k}\right) \sum_{s=1}^{K}\sum_{\mathcal{S}   \subseteq \mathcal{K} \setminus \underline{\mathcal{K}}(\mathbf d): |\mathcal{S}|=s } \max_{k \in \mathcal{S}}y_{d_k,s-1},  \label{eqn:simp_step1_1}
\end{align}
where (a) is due to \eqref{eqn:symmetry}, and (b) is due to Lemma~3 in~\cite{jin}.
Then, by using the same simplification  method in  the proof of Proposition 1 in~\cite{Wei_Yu}, we \textcolor{black}{further} simplify the \textcolor{black}{second} term  in \eqref{eqn:simp_step1_1} into:
\begin{align}
&\sum_{\mathbf{d} \in \mathcal{N}^K} \left(\prod_{k=1}^{K}p_{d_k}\right) \sum_{s=1}^{K-|\underline{\mathcal{K}}(\mathbf d)|}
\sum_{i=1}^{K-|\underline{\mathcal{K}}(\mathbf d)|}{K-|\underline{\mathcal{K}}(\mathbf d)|-i \choose s-1}    y_{d_{k_i},s-1}\nonumber \\
=& \sum_{u=1}^{\min\{K,N\}}\sum_{\mathbf{d} \in \mathcal D_u} \left(\prod_{k=1}^{K}p_{d_k}\right) \sum_{s=1}^{K-u}
\sum_{i=1}^{K-u-s+1}{K-u-i \choose s-1}   y_{d_{k_i},s-1} \nonumber\\
=& \sum_{u=1}^{\min\{K,N\}} \sum_{s=1}^{K-u}
\sum_{i=1}^{K-u-s+1}{K-u-i \choose s-1}\sum_{\mathbf{d} \in \mathcal D_u} \left(\prod_{k=1}^{K}p_{d_k}\right)   y_{d_{k_i},s-1} \nonumber\\
=& \sum_{u=1}^{\min\{K,N\}} \sum_{s=1}^{K-u}
\sum_{i=1}^{K-u-s+1}{K-u-i \choose s-1}\sum_{n=1}^{N} P'_{i,u,n} y_{n,s-1},  \nonumber
\end{align}
where $d_{k_i}$ denotes  the $i$-th most popular file in $\left(D_k\right)_{k \in \mathcal K \setminus \underline{\mathcal{K}}(\mathbf d)}$, and
$\mathcal D_u \triangleq \left\{\mathbf{d} \in \mathcal N^K: \sum_{n=1}^{N} \mathbf 1 [d_n>0]=u\right\}.$

\textcolor{black}{It remains to}  derive  $P'_{i,u,n}$ by connecting the event $\widetilde{D}_{u,\langle i\rangle}=n$ to the ``balls into bins" problem.
\textcolor{black}{Note that the event} that $K$ balls are placed \textcolor{black}{in an} i.i.d. \textcolor{black}{manner} into $N$ bins   \textcolor{black}{where each of the $K$ balls is placed} into bin $n'$ (the bin with index $n'$) with probability  $p_{n'}$  corresponds to the \textcolor{black}{event} that  each user randomly and independently requests  file $n' \in \mathcal N$ with probability $p_{n'}$ \textcolor{black}{(represented by random variable $D_k, k \in \mathcal{K}$)}.
\textcolor{black}{Let $\mathcal{E}_u$ denote the event that there are exactly   $u$} nonempty bins, \textcolor{black}{which corresponds to the event that there are} $u$ representative users, i.e., $|\underline{\mathcal{K}}(\mathbf D)|=u$.
\textcolor{black}{Let $\mathcal{E}^{1, \{1,2, \ldots, n-1\}}_{b_1, a, u}$ denote the event that  $b_1$ balls fall into  $a$  different bins with indices smaller than or equal to $n-1$,    which corresponds to the event that there are $b_1$ users requesting  $a$ different files with file indices smaller than or equal to $n-1$.
Let $\mathcal{E}^{2, \{n+1,n+2, \ldots, N\}}_{b_2, u-a-1, u}$ denote the event that   $b_2$ balls fall into  $u-a-1$ different bins with indices larger than or equal to $n+1$, which corresponds to the event that there are $b_2$ users requesting  $u-a-1$ different files with file indices larger than or equal to $n+1$.
Let $\mathcal{E}^{3, \{n\}}_{b_3, 1, u}$ denote the event that $b_3$ balls fall into bin $n$, which corresponds to the event that there are $b_3$ users requesting   file $n$.
Let \textcolor{black}{$\Theta_{a}(u,n)$} denote the \textcolor{black}{range} of $a$  \textcolor{black}{in $\mathcal{E}^{1, \{1,2, \ldots, n-1\}}_{b_1, a, u}$ and $\mathcal{E}^{2, \{n+1,n+2, \ldots, N\}}_{b_2, u-a-1, u}$}  and let \textcolor{black}{$\Theta_b(n,u, a)$} denote the range of $(b_1, b_2,b_3)$ \textcolor{black}{in $\mathcal{E}^{1, \{1,2, \ldots, n-1\}}_{b_1, a, u}$, $\mathcal{E}^{2, \{n+1,n+2, \ldots, N\}}_{b_2, u-a-1, u}$ and $\mathcal{E}^{3, \{n\}}_{b_3, 1, u}$} for a given $a$, where
\begin{align}
&\textcolor{black}{\Theta_{a}(u,n)}=\begin{cases}
\{0,1, \ldots, u-1\}, & u \leq n, u+n \leq N+1 \\
\{u-1+n-N,  \ldots, u-1\}, & u \leq n,  u+n > N+1\\
\{0, \ldots, n-1\}, & u > n,  u+n \leq N+1\\
\{u-1+n-N, \ldots, n-1\}, & u > n,  u+n > N+1
\end{cases},\label{eqn:Theta_a}
\end{align}
\textcolor{black}{and
\small{\begin{align}
&\textcolor{black}{\Theta_b(n,u, a)}=\begin{cases}
\{(b_1, b_2, b_3): b_1 \in \{a,   \ldots, K\}, b_2 \in  \{u-a-1,  \ldots, K\}, b_1+b_2+b_3=K\}, & a \in \{1, \ldots, u-2\} \cap \Theta_{a}(u,n)\\
\{(b_1, b_2, b_3): b_1 =0, b_2 \in  \{u-1,   \ldots, K\}, b_1+b_2+b_3=K\}, & a \in \{0\} \cap \Theta_{a}(u,n)\\
\{(b_1, b_2, b_3): b_1 \in \{u-1,  \ldots, K\}, b_2 =0, b_1+b_2+b_3=K\}, & a \in \{u-1\} \cap  \Theta_{a}(u,n).
\end{cases} \label{eqn:Theta_b}
\end{align}}}
We know that
\begin{align}
\mathcal{E}_u=\bigcup_{a \in \Theta_{a}(u,n)}\bigcup_{(b_1,b_2,b_3) \in \Theta_b(a)}\left(\mathcal{E}^{1, \{1,2, \ldots, n-1\}}_{b_1, a, u} \cap \mathcal{E}^{2, \{n+1,n+2, \ldots, N\}}_{b_2, u-a-1, u} \cap \mathcal{E}^{3, \{n\}}_{b_3, 1, u}\right). \label{eqn:E_u}
\end{align}}
\textcolor{black}{Then, remove  one ball from each of the $|\underline{\mathcal{K}}(\mathbf D)|$ nonempty bins} \textcolor{black}{and consider the remaining balls, which corresponds to consider  the requests in $(D_k)_{k \in \mathcal K \setminus \underline{\mathcal{K}}(\mathbf D)}$.}
\textcolor{black}{Let $\xi^{1, \{1,2,\ldots, n-1\}}_{i-1}$ denote the event that   there are at most $i-1$ balls placed in the bins \textcolor{black}{with} indices  smaller than or equal to $n-1$
and let $\xi^{2, \{1,2,\ldots, n\}}_{i}$ denote the event that  there are at least $i$ balls placed in the bins \textcolor{black}{with} indices  smaller than or equal to $n$.}
\textcolor{black}{Note that   $\xi^{1, \{1,2,\ldots, n-1\}}_{i-1} \cap \xi^{2, \{1,2,\ldots, n\}}_{i} \cap \mathcal{E}_u$  \textcolor{black}{is equivalent to} the event $\widetilde{D}_{u,\langle i\rangle}=n$,
implying that  $P'_{i,u,n} = \Pr\left[\widetilde{D}_{u,\langle i\rangle}=n\right]=\Pr\left[\xi^{1, \{1,2,\ldots, n-1\}}_{i-1} \cap \xi^{2, \{1,2,\ldots, n\}}_{i}\cap \mathcal{E}_u\right]$.}
\textcolor{black}{\textcolor{black}{Similar to \eqref{eqn:E_u}, $\xi^{1, \{1,2,\ldots, n-1\}}_{i-1} \cap \xi^{2, \{1,2,\ldots, n\}}_{i} \cap \mathcal{E}_u$   can be represented    as
\begin{align}
&\xi^{1, \{1,2,\ldots, n-1\}}_{i-1} \cap \xi^{2, \{1,2,\ldots, n\}}_{i} \cap \mathcal{E}_u \nonumber \\
=&\bigcup_{a \in \Theta_{a}(u,n)}\bigcup_{(b_1,b_2,b_3) \in \widetilde{\Theta}_b(i,u,n,a)}\left(\mathcal{E}^{1, \{1,2, \ldots, n-1\}}_{b_1, a, u} \cap \mathcal{E}^{2, \{n+1,n+2, \ldots, N\}}_{b_2, u-a-1, u} \cap \mathcal{E}^{3, \{n\}}_{b_3, 1, u}\right), \label{eqn:xi_E_u}
\end{align}
where $\widetilde{\Theta}_b(i,u,n,a)$ denotes the range of $(b_1,b_2,b_3)$.
In the following, we determine $\widetilde{\Theta}_b(i,u,n,a)$:}
\begin{itemize}
\item For any $a \in \{1, \ldots, u-2\} \cap \textcolor{black}{\Theta_{a}(u,n)}$,   by   $\mathcal{E}^{1, \{1,2, \ldots, n-1\}}_{b_1, a, u}$ and $\xi^{1, \{1,2,\ldots, n-1\}}_{i-1}$, we have
   \begin{align}
   b_1 \leq  i+a-1; \label{eqn:b_1_1}
   \end{align}
      by $\mathcal{E}^{2, \{n+1,n+2, \ldots, N\}}_{b_2, u-a-1, u}$ and $\xi^{2, \{1,2,\ldots, n\}}_{i}$, we have
\begin{align}
b_2 \leq K-i-a-1.\label{eqn:b_2_1}
\end{align}
By \eqref{eqn:b_1_1}, \eqref{eqn:b_2_1}  and \eqref{eqn:Theta_b}, for  all $a \in \{1, \ldots, u-2\} \cap \textcolor{black}{\Theta_{a}(u,n)}$,  we have
\begin{align}
&\widetilde{\Theta}_b(i,u,n,a) \nonumber \\
=&\left\{(b_1,b_2,b_3):b_1 \in \{a, \ldots, i+a-1\}, b_2 \in \{u-a-1, \ldots, K-i-a-1\}, b_1+b_2 +b_3= K\right\}. \label{eqn:Theta_b_1}
\end{align}
\item For $a \in \{0\} \cap \textcolor{black}{\Theta_{a}(u,n)}$,  by $\mathcal{E}^{1, \{1,2, \ldots, n-1\}}_{b_1, a, u}$ and $\xi^{1, \{1,2,\ldots, n-1\}}_{i-1}$, we have
    \begin{align}
    i=1, b_1=0,\label{eqn:b_1_2}
    \end{align}
     implying that $\mathcal{E}^{1, \{1,2, \ldots, n-1\}}_{b_1, a, u}$ does not happen;     by $\mathcal{E}^{2, \{n+1,n+2, \ldots, N\}}_{b_2, u-a-1, u}$ and $\xi^{2, \{1,2,\ldots, n\}}_{i}$, we have
      \begin{align}
      b_2 \leq  K-i-a-1. \label{eqn:b_2_2}
      \end{align}
      By \eqref{eqn:b_1_2}, \eqref{eqn:b_2_2} and \eqref{eqn:Theta_b}, for all  $a \in \{0\} \cap \textcolor{black}{\Theta_{a}(u,n)}$, we have
      \begin{align}
      \widetilde{\Theta}_b(i,u,n,a)=\left\{(b_1, b_2,b_3):b_1=0,  b_2 \in \{u-1, \ldots, K-a-2\},  b_1+b_2 +b_3= K\right\}.  \label{eqn:Theta_b_2}
      \end{align}
\item   For $a \in \{u-1\} \cap \textcolor{black}{\Theta_{a}(u,n)}$, by $\mathcal{E}^{1, \{1,2, \ldots, n-1\}}_{b_1, a, u}$ and $\xi^{1, \{1,2,\ldots, n-1\}}_{i-1}$, we have
    \begin{align}
    i=K-u, b_2=0,\label{eqn:b_1_3}
    \end{align}
     implying that $\mathcal{E}^{2, \{n+1,n+2, \ldots, N\}}_{b_2, u-a-1, u}$ does not happen;    by $\mathcal{E}^{2, \{n+1,n+2, \ldots, N\}}_{b_2, u-a-1, u}$ and $\xi^{2, \{1,2,\ldots, n\}}_{i}$, we have
     \begin{align}
     b_1 \leq  a+i-1.\label{eqn:b_2_3}
     \end{align}
     By \eqref{eqn:b_1_3}, \eqref{eqn:b_2_3} and \eqref{eqn:Theta_b}, for all  $a \in \{u-1\} \cap \textcolor{black}{\Theta_{a}(u,n)}$, we have
     \begin{align}
     \widetilde{\Theta}_b(i,u,n,a)=\left\{(b_1, b_2,b_3): b_1 \in \{u-1, \ldots, K-u+a-1 \}, b_2 =0,  b_1+b_2 +b_3= K\right\}.  \label{eqn:Theta_b_3}
     \end{align}
\end{itemize}
By \eqref{eqn:xi_E_u}, \eqref{eqn:Theta_b_1},  \eqref{eqn:Theta_b_2} and \eqref{eqn:Theta_b_3},  we have
\begin{align}
&\xi^{1, \{1,2,\ldots, n-1\}}_{i-1} \cap \xi^{2, \{1,2,\ldots, n\}}_{i} \cap \mathcal{E}_u \nonumber\\
=&\Bigg(\bigcup_{a \in \{1, \ldots, u-2\} \cap \textcolor{black}{\Theta_{a}(u,n)}}\bigcup_{(b_1,b_2,b_3) \in \widetilde{\Theta}_b(i,u,n,a)}
\left(\mathcal{E}^{1, \{1,2, \ldots, n-1\}}_{b_1, a, u}\cap \mathcal{E}^{2, \{n+1,n+2, \ldots, N\}}_{b_2, u-a-1, u}\cap \mathcal{E}^{3, \{n\}}_{K-b_1-b_2, 1, u}\right)\Bigg)\nonumber\\
\bigcup &\left(\bigcup_{a \in \{0\} \cap \textcolor{black}{\Theta_{a}(u,n)}}  \bigcup_{(b_1,b_2,b_3) \in \widetilde{\Theta}_b(i,u,n,a)} \left(\mathcal{E}^{2, \{n+1,n+2, \ldots, N\}}_{b_2, u-a-1, u}\cap \mathcal{E}^{3, \{n\}}_{K-b_2, 1, u}\right)\right)\nonumber
\end{align}
\begin{align}
\bigcup &\left(\bigcup_{a \in \{u-1\} \cap \textcolor{black}{\Theta_{a}(u,n)}} \bigcup_{(b_1,b_2,b_3) \in \widetilde{\Theta}_b(i,u,n,a)}\left(\mathcal{E}^{1, \{1,2, \ldots, n-1\}}_{b_1, a, u}\cap \mathcal{E}^{3, \{n\}}_{K-b_1, 1, u}\right)\right). \label{eqn:convert_xi_e}
\end{align}
Based on \eqref{eqn:convert_xi_e}, we have
\begin{align}
P'_{i,u,n}&=\sum_{a \in \{1, \ldots, u-2\} \cap \textcolor{black}{\Theta_{a}(u,n)}}\sum_{b_1 \in \{a, \ldots, i+a-1\}}\sum_{b_2 \in \{u-a-1, \ldots, K-i-a-1\}}{K \choose b_1,K-b_1-b_2, b_2}  \nonumber \\
& \times \Pr\left[\mathcal{E}^{1, \{1,2, \ldots, n-1\}}_{b_1, a, u} \cap \mathcal{E}^{2, \{n+1,n+2, \ldots, N\}}_{b_2, u-a-1, u} \cap \mathcal{E}^{3, \{n\}}_{K-b_1-b_2, 1, u}\right] \nonumber \\
& + \sum_{a \in \{0\} \cap \textcolor{black}{\Theta_{a}(u,n)}} \sum_{b_2 \in \{u-1, \ldots, K-a-2\}} {K \choose  b_2}  \Pr\left[ \mathcal{E}^{2, \{n+1,n+2, \ldots, N\}}_{b_2, u-a-1, u} \cap \mathcal{E}^{3, \{n\}}_{K-b_2, 1, u}\right]\nonumber \\
& + \sum_{a \in \{u-1\} \cap \textcolor{black}{\Theta_{a}(u,n)}} \sum_{b_1 \in \{u-1, \ldots, K-u+a-1 \}} {K \choose  b_1}   \Pr\left[\mathcal{E}^{1, \{1,2, \ldots, n-1\}}_{b_1, a, u} \cap \mathcal{E}^{3, \{n\}}_{K-b_1, 1, u}\right], \label{eqn:3_cond_1}
\end{align}
where ${K \choose b_1,K-b_1-b_2, b_2}$
is the total number  of partitions of  $K$ balls into three parts   with the numbers of balls  $b_1$, $b_2$ and $K-b_1-b_2$,
${K \choose  b_2}$
is the total number  of partitions of  $K$ balls into two parts  with the numbers of balls $b_2$ and $K-b_2$,
and ${K \choose  b_1}$
is the total number  of partitions of  $K$ balls into two parts  with the  numbers of balls $b_1$ and $K-b_1$.
To calculate \eqref{eqn:3_cond_1}, we first calculate $\Pr\left[\mathcal{E}^{1, \{1,2, \ldots, n-1\}}_{b_1, a, u} \cap \mathcal{E}^{2, \{n+1,n+2, \ldots, N\}}_{b_2, u-a-1, u}  \cap \mathcal{E}^{3, \{n\}}_{K-b_1-b_2, 1, u}\right]$ for \textcolor{black}{all} $a \in \{1, \ldots, u-2\} \cap \textcolor{black}{\Theta_{a}(u,n)}$, $b_1 \in \{a, a+1, \ldots, i+a-1\}$ and $b_2\in  \{u-a-1,u-a, \ldots, K-i-a-1\}$.
By using  results from ``balls into bins" problem,
we have
\begin{align}
&\Pr\left[\mathcal{E}^{1, \{1,2, \ldots, n-1\}}_{b_1, a, u} \cap \mathcal{E}^{2, \{n+1,n+2, \ldots, N\}}_{b_2, u-a-1, u} \cap \mathcal{E}^{3, \{n\}}_{K-b_1-b_2, 1, u}\right] \nonumber \\
=&\Pr \left[\mathcal{E}^{1, \{1,2, \ldots, n-1\}}_{b_1, a, u}\right] \Pr \left[\mathcal{E}^{2, \{n+1,n+2, \ldots, N\}}_{b_2, u-a-1, u}\big| \mathcal{E}^{1, \{1,2, \ldots, n-1\}}_{b_1, a, u}\right]   \nonumber \\
& \times \Pr \left[\mathcal{E}^{3, \{n\}}_{K-b_1-b_2, 1, u}\big|\mathcal{E}^{1, \{1,2, \ldots, n-1\}}_{b_1, a, u} \cap \mathcal{E}^{2, \{n+1,n+2, \ldots, N\}}_{b_2, u-a-1, u}\right], \label{eqn:3_cond_1_1}
\end{align}
where
$\Pr \left[\mathcal{E}^{1, \{1,2, \ldots, n-1\}}_{b_1, a, u}\right]=\sum_{\mathcal L_1 \in \mathcal L_{a,1}}\sum_{(\alpha_{n'})_{n' \in \mathcal L_1} \in \mathcal A_{b_1, \mathcal L_1}} \frac{b_1!}{\prod_{n' \in \mathcal L_1}\alpha_{n'}!}\prod_{n' \in \mathcal L_1}P_{n'}^{\alpha_{n'}},$
$$\Pr \left[\mathcal{E}^{2, \{n+1,n+2, \ldots, N\}}_{b_2, u-a-1, u}\big|\mathcal{E}^{1, \{1,2, \ldots, n-1\}}_{b_1, a, u}\right]=\sum_{\mathcal L_2 \in \mathcal L_{a,2}}\sum_{(\alpha_{n'})_{n' \in \mathcal L_2} \in  \mathcal A_{b_2, \mathcal L_2}}  \frac{b_2!}{\prod_{n' \in \mathcal L_2}\alpha_{n'}!}\prod_{n' \in \mathcal L_2}P_{n'}^{\alpha_{n'}},$$
$\Pr \left[\mathcal{E}^{3, \{n\}}_{K-b_1-b_2, 1, u}\big|  \mathcal{E}^{1, \{1,2, \ldots, n-1\}}_{b_1, a, u} \cap \mathcal{E}^{2, \{n+1,n+2, \ldots, N\}}_{b_2, u-a-1, u}\right]=P_n^{K-b_1-b_2}.$
Next, we calculate
$$\Pr\left[\mathcal{E}^{2, \{n+1,n+2, \ldots, N\}}_{b_2, u-a-1, u}, \mathcal{E}^{3, \{n\}}_{K-b_2, 1, u}\right]$$ for \textcolor{black}{all} $a \in \{0\} \cap \textcolor{black}{\Theta_{a}(u,n)}$,  $b_1=0$, \textcolor{black}{and} $b_2\in  \{u-1,  \ldots, K-a-i-1\}$.
By using  results from ``balls into bins" problem, we have
\begin{align}
\Pr\left[\mathcal{E}^{2, \{n+1,n+2, \ldots, N\}}_{b_2, u-a-1, u} \cap \mathcal{E}^{3, \{n\}}_{K-b_2, 1, u}\right]= \Pr \left[\mathcal{E}^{2, \{n+1,n+2, \ldots, N\}}_{b_2, u-a-1, u}\right]
\Pr \left[\mathcal{E}^{3, \{n\}}_{K-b_2, 1, u}\big| \mathcal{E}^{2, \{n+1,n+2, \ldots, N\}}_{b_2, u-a-1, u}\right], \label{eqn:3_cond_1_2}
\end{align}
where
$ \Pr \left[\mathcal{E}^{2, \{n+1,n+2, \ldots, N\}}_{b_2, u-a-1, u}\right]=\sum_{\mathcal L_2 \in \mathcal L_{0,2}}\sum_{(\alpha_{n'})_{n' \in \mathcal L_2} \in  \mathcal A_{b_2, \mathcal L_2}} \frac{b_2!}{\prod_{n' \in \mathcal L_2}\alpha_{n'}!}\prod_{n' \in \mathcal L_2}P_{n'}^{\alpha_{n'}},$
and $$\Pr \left[\mathcal{E}^{3, \{n\}}_{K-b_2, 1, u}\big|  \mathcal{E}^{2, \{n+1,n+2, \ldots, N\}}_{b_2, u-a-1, u}\right]=P_n^{K-b_2}.$$
Finally, we calculate $\Pr\left[\mathcal{E}^{1, \{1,2, \ldots, n-1\}}_{b_1, a, u} \cap \mathcal{E}^{3, \{n\}}_{K-b_1, 1, u}\right]$ for \textcolor{black}{all} $a \in \{u-1\} \cap \textcolor{black}{\Theta_{a}(u,n)}$,   $b_2=0$ and $b_1 \in \{u-1, \ldots, a+i-1\}$.
By using  results from ``balls into bins" problem, we have
\begin{align}
\Pr\left[\mathcal{E}^{1, \{1,2, \ldots, n-1\}}_{b_1, a, u} \cap \mathcal{E}^{3, \{n\}}_{K-b_1, 1, u}\right]=\Pr \left[\mathcal{E}^{1, \{1,2, \ldots, n-1\}}_{b_1, a, u}\right]
\Pr \left[\mathcal{E}^{3, \{n\}}_{K-b_1, 1, u}\big| \mathcal{E}^{1, \{1,2, \ldots, n-1\}}_{b_1, a, u}\right], \label{eqn:3_cond_1_3}
\end{align}
where $\Pr \left[\mathcal{E}^{3, \{n\}}_{K-b_1, 1, u}\big| \mathcal{E}^{1, \{1,2, \ldots, n-1\}}_{b_1, a, u}\right]=P_n^{K-b_1},$
$$\Pr \left[\mathcal{E}^{1, \{1,2, \ldots, n-1\}}_{b_1, a, u}\right]=\sum_{\mathcal L_1 \in \mathcal L_{u-1,1}}\sum_{(\alpha_{n'})_{n' \in \mathcal L_1} \in  \mathcal A_{b_1, \mathcal L_1}}  \frac{b_1!}{\prod_{n' \in \mathcal L_1}\alpha_{n'}!}\prod_{n' \in \mathcal L_1}P_{n'}^{\alpha_{n'}}.$$
Substituting \eqref{eqn:3_cond_1_1}, \eqref{eqn:3_cond_1_2} and \eqref{eqn:3_cond_1_3} into \eqref{eqn:3_cond_1}, we can obtain $P'_{i,u,n}$   given in \eqref{eqn:region1}-\eqref{eqn:region4}.
Therefore, we complete the proof of Lemma~\ref{Lem:simplification pop}.}

\section*{Appendix B: Proof of Lemma~\ref{Lem:uniform_obj_simp}}
By~\eqref{eqn:symmetry} and  \eqref{eqn:symmetry_2}, it can be easily shown that the constraints in \eqref{eqn:X_range}, \eqref{eqn:X_sum} and \eqref{eqn:memory_constraint} of Problem~\ref{Prob:original} can be converted into \eqref{eqn:X_range_3}, \eqref{eqn:X_sum_3} and \eqref{eqn:memory_constraint_3}.
\textcolor{black}{Now, we}  show that the objective
function of Problem~\ref{Prob:original} in \eqref{eqn:average_load_1} can be converted into the objective function of Problem~\ref{Prob:equivalent_3}.
By \eqref{eqn:symmetry} and \eqref{eqn:symmetry_2}, we have
\begin{align}
&R_{\rm avg}(K,N,M,\mathbf{x})
\overset{(a)}= \sum_{\mathbf{d} \in \mathcal{N}^K} \left(\prod_{k=1}^{K}p_{d_k}\right) \sum_{s=1}^{K}\sum_{\mathcal{S}   \subseteq \mathcal{K}: |\mathcal{S}|=s} z_{s-1} -\sum_{\mathbf{d} \in \mathcal{N}^K} \left(\prod_{k=1}^{K}p_{d_k}\right) \sum_{s=1}^{K}\sum_{\mathcal{S}   \subseteq \mathcal{K} \setminus \underline{\mathcal{K}}(\mathbf d): |\mathcal{S}|=s }   z_{s-1}  \nonumber \\
=&\sum_{s=1}^{K}{K \choose s}z_{s-1}- \sum_{u=1}^{\min \{K,N\}} \Pr \left[|\underline{\mathcal{K}}(\mathbf D)|=u\right] \sum_{s=1}^{K}{K-u \choose s} z_{s-1}
\end{align}
where (a) is due to \eqref{eqn:symmetry} and \eqref{eqn:symmetry_2}.
Therefore, we complete the proof of Lemma~\ref{Lem:uniform_obj_simp}.

\section*{Appendix C: Proof of Lemma~\ref{Lem:Ali}}
The Lagrangian of Problem~\ref{Prob:equivalent_3}  is given by
$
L(\mathbf{z}, \boldsymbol \eta, \theta, \nu)
=\sum_{s=0}^{K-1}{K \choose s+1}  z_{s} + \sum_{s=0}^{K} \eta_{s} \left(-z_{s}\right) - \sum_{u=1}^{\min \{K,N\}} P''_{u}
\sum_{s=0}^{K-u-1}{K-u \choose s+1} z_{s} +\theta  \left(\sum_{s=0}^{K} {K \choose s}s z_{s}-\frac{KM}{N}\right)
+  \nu \left(1-\sum_{s=0}^{K} {K \choose s} z_{s}\right),
$
where $\eta_{s} \geq 0$ is the Lagrange multiplier associated with~\eqref{eqn:X_range_3}, $\nu$ is the Lagrange multiplier associated with~\eqref{eqn:X_sum_3}, $\theta$ is the Lagrange multiplier associated with~\eqref{eqn:memory_constraint_3} and $\boldsymbol \eta \triangleq (\eta_{s})_{s \in \{0,1,\ldots,K\}}$. Thus, we have
\begin{align}
\frac{\partial L}{\partial z_{s}}(\mathbf{z}, \boldsymbol \eta,  \theta, \nu)={K \choose s+1}- \sum_{u=1}^{\min \{K,N\}} P''_{u}{K-u \choose s+1}-\eta_{s}+\theta s{K \choose s}-\nu {K \choose s}.
\end{align}
Since Problem~\ref{Prob:equivalent_3} is a linear programming,  $\mathbf{z}^*$ is an optimal solution of Problem~\ref{Prob:equivalent_3} if  $\mathbf{z}^*$,  $\boldsymbol \eta^*, \nu^*, \theta^*$ satisfy KKT conditions, i.e., (i) primal constraints: \eqref{eqn:X_range_3}, \eqref{eqn:X_sum_3}, \eqref{eqn:memory_constraint_3}, (ii) dual constraints: (a) $\theta \geq 0$  and (b) $\eta_{s} \geq 0$    for all $s \in \{0,1,\ldots,K\}$,
(iii) complementary slackness: (a) $\eta_{s} \left(-z_{s}\right)=0$  for all $s \in \{0,1,\ldots,K\}$ and (b) $\theta  \left(\sum_{s=0}^{K} {K \choose s}s z_{s}-\frac{KM}{N}\right)=0$,  and (iv) $\frac{\partial L}{\partial z_{s}}(\mathbf{z}, \boldsymbol \eta,  \theta, \nu)=0$  for all $s \in \{0,1,\ldots,K\}$.
\textcolor{black}{In the following, we obtain $\mathbf{z}^*$,  $\boldsymbol \eta^*$, $\nu^*$, and $\theta^*$ by considering three cases.}

\textcolor{black}{Case 1: When} $M=0$, it can be easily verified that $\mathbf z^*$ given in  \eqref{eqn:uniform_caching}, any
$\nu^*=\sum_{u=1}^{\min \{K,N\}} P''_{u}u,$
$\eta^*_{s}=\sum_{u=1}^{\min \{K,N\}} P''_{u} \left(-u{K \choose s}-{K-u \choose s+1}+{K \choose s+1}\right)+\theta^*s {K \choose s}$,
$$\theta^* \in \left[ \max\left\{\max_{s \in \{1,2, \ldots, K\}}\left\{\sum_{u=1}^{\min \{K,N\}} P''_{u}\frac{u{K \choose s}+{K-u \choose s+1}-{K \choose s+1}}{s {K \choose s}}\right\},0\right\},+\infty\right),$$
satisfy the KKT conditions in (i)-(iv).
Thus, we know that \textcolor{black}{when}  $M=0$, $\mathbf{z}^*$ given in \eqref{eqn:uniform_caching} is an optimal solution of Problem~\ref{Prob:equivalent_3}.

\textcolor{black}{Case 2: When} $M=N$, it can be easily verified that $\mathbf z^*$ given in  \eqref{eqn:uniform_caching}, any $\nu^*=K\theta^*, $
$\theta^* \in \left[0, \min_{s \in \{0,1, \ldots, K-1\}}\left\{\sum_{u=1}^{\min \{K,N\}} P''_{u}\frac{u{K \choose s}+{K-u \choose s+1}-{K \choose s+1}}{s {K \choose s}}\right\}\right],$
$\eta^*_{s}={K \choose s+1}-\sum_{u=1}^{\min \{K,N\}} P''_{u} {K-u \choose s+1} +\theta^*(s-K) {K \choose s}$
satisfy the KKT conditions in (i)-(iv).
Thus, we know that \textcolor{black}{when} $M=N$, $\mathbf{z}^*$ given in \eqref{eqn:uniform_caching} is an optimal solution of Problem~\ref{Prob:equivalent_3}.

\textcolor{black}{Case 3: When}  any    $M \in \left\{ \frac{N}{K}, \frac{2N}{K}, \ldots, \frac{(K-1)N}{K}\right\}$, we prove that $\mathbf{z}^*$ given in \eqref{eqn:uniform_caching} is an   optimal solution of Problem~\ref{Prob:equivalent_3}  by proving that
$\mathbf z^*$ given in  \eqref{eqn:uniform_caching},
\begin{align}
\theta^* =-\sum_{u=1}^{\min \{K,N\}} P''_{u} \sum_{k=K-u+1}^{K}g'_k\left(\frac{KM}{N}\right), \label{eqn:theta}
\end{align}
\begin{align}
\nu^* = \sum_{u=1}^{\min \{K,N\}} P''_{u} \sum_{k=K-u+1}^{K}\left(g_k\left(\frac{KM}{N}\right)-g'_k\left(\frac{KM}{N}\right)\frac{KM}{N}\right),\label{eqn:nu}
\end{align}
\begin{align}
\eta^*_{s}= {K \choose s}\sum_{u=1}^{\min \{K,N\}} P''_{u}\sum_{k=K-u+1}^{K}\left(g_k(s)-\left(g_k\left(\frac{KM}{N}\right)+g'_k\left(\frac{KM}{N}\right)\left(s-\frac{KM}{N}\right)\right)\right),\nonumber \\
 s \in \{0,1,\ldots,K\},\label{eqn:eta}
\end{align}
satisfy the KKT conditions in (i)-(iv), where
\begin{align}
g_k(s) \triangleq
\begin{cases}
h_k(s) \triangleq  \frac{\prod_{i=0}^{K-k}(K-s-i)}{\prod_{j=0}^{K-k}(K-j)}, & s \in [0,k)\\
0, &  s \in [k,K]
\end{cases},
\ k \in \{K-u+1, \ldots, K\}.
\end{align}
\textcolor{black}{In the following, we show that $\mathbf{z}^*$  in \eqref{eqn:uniform_caching}, $\theta^*$ in \eqref{eqn:theta},  $\nu^*$ in \eqref{eqn:nu}, and $\boldsymbol \eta^*$  in \eqref{eqn:eta} satisfy KKT conditions (i), (ii), (iii), and (iv), respectively.}
\begin{itemize}
\item  Prove that  $\mathbf{z}^*$ satisfies (i).
By substituting $\mathbf{z}^*$  into the primal constraints in \eqref{eqn:X_range_3}, \eqref{eqn:X_sum_3} and \eqref{eqn:memory_constraint_3},  we can easily verify that $\mathbf{z}^*$ satisfies (i).
    Thus, we complete proving  that $\mathbf{z}^*$ satisfies (i).
\item  Prove that $\theta^*$ satisfies (ii.a) and $\boldsymbol \eta^*$ satisfies  (ii.b).
\textcolor{black}{First, we show that $\theta^*$ satisfies (ii.a)}  by proving $g'_k\left(\frac{KM}{N}\right) \leq 0$. Since
\begin{align}
g'_k\left(s\right) =
\begin{cases}
h'_k(s)= -\frac{\prod_{i=0}^{K-k}(K-s-i)}{\prod_{j=0}^{K-k}(K-j)}\sum_{i=0}^{K-k}\frac{1}{K-s-i}, & s \in (0,k)\\
0, & s \in [k,K)
\end{cases},
\end{align}
we have $g'_k\left(s\right) \leq 0$ for any $s \in (0,K)$.
Since $M \in \left\{ \frac{N}{K}, \frac{2N}{K}, \ldots, \frac{(K-1)N}{K}\right\}$ , we have $\frac{KM}{N} \in \{1,2,\ldots, K-1\}$.
Therefore, we have  $g'_k\left(\frac{KM}{N}\right) \leq 0$.
Thus, we complete proving that $\theta^*$ satisfies (ii.a).
\textcolor{black}{Next, we show that $\boldsymbol \eta^*$ satisfies  (ii.b)} by proving
\begin{align}
g_k(s)-\left(g_k\left(\frac{KM}{N}\right)+g'_k\left(\frac{KM}{N}\right)\left(s-\frac{KM}{N}\right)\right) \geq 0. \label{eqn:diff_g}
\end{align}
\textcolor{black}{Consider the following four cases.}

1) When $s \in [0,k)$ and $\frac{KM}{N} \in (0,k)$, \eqref{eqn:diff_g} is equivalent to
\begin{align}
h_k(s) \geq  h_k\left(\frac{KM}{N}\right)+h'_k\left(\frac{KM}{N}\right)\left(s-\frac{KM}{N}\right). \label{eqn:diff_h}
\end{align}
Since for any \textcolor{black}{convex set $\mathcal{X}$ and any two points $x, y \in \mathcal{X}$,}  $f(y) \geq f(x) + f'(x)(y-x)$ if and only if $f(x)$ is convex~\cite{Boyd},
we prove \eqref{eqn:diff_h}  by proving that $h_k(s)$ is convex over $s \in [0,k)$.
Since $h''_k(s) =h_k(s) \left(\left(\sum_{i=0}^{K-k}\frac{1}{K-s-i}\right)^2-\sum_{i=0}^{K-k}\frac{1}{\left(K-s-i\right)^2}\right) \geq 0$ for all $s \in [0,k)$, by second-order condition for convexity~\cite{Boyd}, we know that  $h_k(s)$ is convex over $s \in [0,k)$.

2) When $s \in [0,k)$ and $\frac{KM}{N} \in [k,K)$, \eqref{eqn:diff_g} is equivalent to $h_k(s) \geq 0$, which  holds for \textcolor{black}{all} $s \in [0,k)$.

3) When $s \in [k,K]$ and $\frac{KM}{N} \in (0,k)$, \eqref{eqn:diff_g} is equivalent to
$\sum_{i=0}^{K-k}\frac{s-\frac{KM}{N}}{K-\frac{KM}{N}-i} \geq 1,$
which always holds since $\sum_{i=0}^{K-k}\frac{s-\frac{KM}{N}}{K-\frac{KM}{N}-i} \geq \sum_{i=0}^{K-k}\frac{k-\frac{KM}{N}}{K-\frac{KM}{N}-i}=1+\sum_{i=0}^{K-k-1}\frac{k-\frac{KM}{N}}{K-\frac{KM}{N}-i} \geq 1$.

4) When $s \in [k,K]$ and $\frac{KM}{N} \in [k,K)$, \eqref{eqn:diff_g} is equivalent to $0 \geq 0$, which always holds.

Thus,  we complete proving that $\boldsymbol \eta^*$ satisfies  (ii.b).

\item  Prove that $\mathbf{z}^*, \boldsymbol \eta^*$ satisfies (iii.a) and $\mathbf{z}^*, \theta^*$ satisfies (iii.b).
\textcolor{black}{First, we show that $\mathbf{z}^*, \boldsymbol \eta^*$ satisfies (iii.a).} Since $\eta^*_{\frac{KM}{N}}=0$ and $z^*_s=0$ for all $s \in \{0,1,\ldots, K\}\setminus \left\{\frac{KM}{N}\right\}$,  we know that $\eta^*_{s} \left(-z_{s}\right)=0$  for all $s \in \{0,1,\ldots,K\}$. Thus, we complete proving that $\mathbf{z}^*, \boldsymbol \eta^*$ satisfies~(iii.a).
\textcolor{black}{Next, we show that $\mathbf{z}^*, \theta^*$ satisfies (iii.b).}  Since  $\sum_{s=0}^{K} {K \choose s}s z^*_{s}=\frac{KM}{N}$, we have
$\theta^*  \left(\sum_{s=0}^{K} {K \choose s}s z^*_{s}-\frac{KM}{N}\right)=0$.
    Thus, we complete proving that $\mathbf{z}^*, \theta^*$ satisfies (iii.b).

\item  Prove that $\mathbf{z}^*, \boldsymbol \eta^*, \nu^*$ and $\theta^*$ satisfies (iv). For any $s \in \{0,1,\ldots,K\}$, we have
$$\frac{\partial L}{\partial z_{s}}(\mathbf{z}^*, \boldsymbol \eta^*,  \theta^*, \nu^*)=\sum_{u=1}^{\min \{K,N\}} P''_{u} \left({K \choose s+1}-{K-u \choose s+1}-\sum_{k=K-u+1}^{K}g_k(s){K \choose s}\right).$$
By Pascal's identity, i.e., ${{k+1} \choose t}={k \choose t}+{k \choose {t-1}}$, we have  ${K \choose s+1}-{K-u \choose s+1}=\sum_{k=K-u+1}^{K} {k-1 \choose s}$.
Furthermore, for any $s \in \{0,1,\ldots,K\}$,  we have $g_k(s)=\frac{{k-1 \choose s}}{{K \choose s}}$.
Therefore, for any $s \in \{0,1,\ldots,K\}$,  we have $\frac{\partial L}{\partial z_{s}}(\mathbf{z}^*, \boldsymbol \eta^*,  \theta^*, \nu^*)=0$.
Thus, we complete proving that $\mathbf{z}^*, \boldsymbol \eta^*, \nu^*$ and $\theta^*$ satisfies (iv).
\end{itemize}
\textcolor{black}{Combining the above three cases}, we know that $\mathbf z^*$,  $\boldsymbol \eta^*$,  $\theta^*$ and  $\nu^*$  satisfy the KKT conditions in (i)-(iv).
Therefore, we complete the proof of Lemma~\ref{Lem:Ali}.

\section*{Appendix D: Proof of Lemma~\ref{Lem:SimpSubpackCons}}
Under   Conditions~\ref{Con:symmetry} and~\ref{Con:popularity},
for the simplified  problem,   the average load,  the file partition constraints,   and  the cache memory constraint are the same as those of Problem~\ref{Prob:simplify_2}.
It remains to  transform  the   subpacketization level  constraint  in~\eqref{eqn:subpack} in terms of  the vector  $\mathbf{x}$ to \eqref{eqn:subpack_eq_y}  in terms of  the vector   $\mathbf{y}$.
First, by Theorem~1 of~\cite{dc2017},  the  subpacketization  level  constraint in~\eqref{eqn:subpack}  is equivalent to
\begin{align}
\| \mathbf{x}_n \|_{lgst, \widehat{F}} \geq \| \mathbf{x}_n \|_1, \ n \in \mathcal N. \label{eqn:DC_x}
\end{align}
By the file partition constraint in \eqref{eqn:X_sum}, we have
\begin{align}
\| \mathbf{x}_n \|_1= \sum_{s=0}^{K}\sum_{\mathcal{S}\in \{\mathcal{\widehat{S}} \subseteq \mathcal{K}: |\mathcal{\widehat{S}}|=s\}}x_{n, \mathcal{S}}=1, \ n \in \mathcal N. \label{eqn:x_norm_N}
\end{align}
\textcolor{black}{By \eqref{eqn:DC_x} and \eqref{eqn:x_norm_N}, we have}
\begin{align}
\| \mathbf{x}_n \|_{lgst, \widehat{F}} \geq 1, \ n \in \mathcal N. \label{eqn:DC_x2}
\end{align}
\textcolor{black}{Next, under Condition~\ref{Con:symmetry}, it is clear that}
$\mathbf{x}_n=U_n\mathbf{y}, \ n \in \mathcal N$.
\textcolor{black}{Thus, by \eqref{eqn:DC_x2}, we can obtain~\eqref{eqn:subpack_eq_y}.}
Therefore, we complete the proof of Lemma~\ref{Lem:SimpSubpackCons}.

\section*{Appendix E: Proof of Lemma~\ref{Lem:subgra}}
Let $\mathbf v$ be an $L$-dimensional vector.
\textcolor{black}{Denote} $\mathbf{f}^m(\mathbf  v) \triangleq (f_j^{m,\mathbf v})_{j \in \left\{1,2,\ldots, L\right\}}$,  where
\begin{align}
&f_{[j]}^{m,\mathbf v}=\begin{cases}
1, & j \in \{1, \ldots, m\} \\
0, &\text{otherwise}
\end{cases}, \ m \in \{1,2, \ldots, L\},
\end{align}
and $[j]$ represents the index of the $j$-th largest element in  $\mathbf{v}$.
\textcolor{black}{By~\cite{dc2017}, we know that} $\mathbf{g}(\mathbf y)  = U_n^T \mathbf{f}^{\widehat{F}}\left(U_n\mathbf{y}\right)$ is a subgradient of $\|U_n\mathbf{y} \|_{lgst, \widehat{F}}$.
In the following, we first calculate $\mathbf{f}^{\widehat{F}}\left(U_n\mathbf{y}\right)$.
Let $y_{n,s_{[i]}}$ denote the $i$-th largest element in $\mathbf{y}_n$, where   $s_{[i]} \triangleq  [i]-1$.
Since $U_n\mathbf{y}$ is equivalent to $\mathbf{x}_n$  under Condition~\ref{Con:symmetry},  the $\left(\sum_{i'=1}^{i-1}{K \choose s_{[i]}}+1\right)$-th to the   $\left(\sum_{i'=1}^{i}{K \choose s_{[i]}}\right)$-th largest elements in $U_n\mathbf{y}$  all equal   to $y_{n,s_{[i]}}$.
Thus, the set of the  indices of the  $\widehat{F}$-th largest  element  in $U_n\mathbf{y}$ is
\begin{align} \label{eqn:omega}
 \Omega^{\widehat{F},U_n\mathbf{y}} = &\underset{{i \in \{1, \ldots, I-1\}}}\bigcup\left\{\sum_{s'=0}^{s_{[i]}-1}{K \choose s'}+1, \ldots, \sum_{s'=0}^{s_{[i]}}{K \choose s'}\right\}\nonumber \\
&\bigcup \left\{\sum_{s'=0}^{s_{[I]}-1}{K \choose s'}+1, \ldots, \sum_{s'=0}^{s_{[I]}-1}{K \choose s'}+\widehat{F}-\sum_{i=1}^{I-1} {K \choose s_{[i]} }\right\}.
\end{align}
Thus, \textcolor{black}{we have} $\mathbf{f}^{\widehat{F}}\left(U_n\mathbf{y}\right)\triangleq (f_j^{\widehat{F},U_n\mathbf{y}})_{j \in \left\{1,2,\ldots, 2^K\right\}}$,  where
\begin{align} \label{eqn:fj}
&f_{j}^{\widehat{F},W_n\mathbf{y}}=\begin{cases}
1, & j \in  \Omega^{\widehat{F},U_n\mathbf{y}} \\
0, &\text{otherwise}
\end{cases}.
\end{align}
Next, we calculate $U_n^T$.
Let $U_n^T \triangleq (u_{m,h})_{m \in \{1,2, \ldots, (K+1)N\},h \in \{1,2, \ldots,  2^K\}}$,   and $l_m \triangleq (m-1) \mod (K+1)$.
From the definition of $U_n$, we know that for any $m \in \{1,2,\ldots, N(K+1)\}$,
\begin{align}  \label{eqn:W_T_2}
u_{m,h} =\begin{cases}
1, & m=n, \ h \in \left\{\sum_{l=0}^{l_m-1}{K \choose l}+1, \ldots, \sum_{l=0}^{l_m}{K \choose l}\right\}\\
0, &\text{otherwise}
\end{cases}.
\end{align}
By  \eqref{eqn:omega}, \eqref{eqn:fj} and \eqref{eqn:W_T_2},  we can obtain   $\mathbf{g}_n(\mathbf y) =U_n^T \mathbf{f}^{\widehat{F}}\left(U_n\mathbf{y}\right)$,   indicating  \eqref{eqn:g_t}.
Therefore, we complete the proof of Lemma~\ref{Lem:subgra}.
\bibliography{reference}
\end{document}